\documentclass[12pt]{article}
\usepackage[margin=25mm]{geometry}

\providecommand{\keywords}[1]
{
  \small	
  \textbf{\textit{Keywords---}} #1
}

\usepackage{amsmath}
\usepackage{amsthm}
\usepackage{amssymb}
\usepackage{graphicx}
\usepackage{subcaption}

\usepackage{hyperref}
\usepackage{enumerate}
\usepackage{natbib}
\usepackage{url} 
\usepackage{threeparttable}
\usepackage{tablefootnote}
\usepackage{booktabs}
\newtheorem{theorem}{Theorem}
\newtheorem{lemma}[theorem]{Lemma}
\newtheorem{assumption}{Assumption}

\begin{document}
\title{Kernel Assisted Learning for Personalized Dose Finding}


%\email{webmaster@marysville-ohio.com}
\author{Liangyu Zhu$^1$\footnote{Email: lzhu12@ncsu.edu},
        Wenbin Lu$^1$, 
        Michael R. Kosorok$^2$,
        Rui Song$^1$\\
        $^1$Department of Statistics, 
        North Carolina State University\\
        $^2$Department of Biostatistics and Department of Statistics, \\University of North Carolina at Chapel Hill}
\date{}

\maketitle
%% Use this when submitting an article to a sponsored event. You'll
%% receive a unique submission ID from the organizers
%% of the event, and this ID should be used as the parameter to this command.
%%\acmSubmissionID{123-A56-BU3}
\begin{abstract}
  An individualized dose rule 
  recommends a dose level within a continuous safe dose range based on patient level information such as physical conditions, genetic factors and medication histories. 
  Traditionally, personalized dose finding process requires repeating clinical visits of the patient and frequent adjustments of the dosage. Thus the patient is constantly exposed to the risk of underdosing and overdosing during the process.
  Statistical methods for finding an optimal individualized dose rule can lower the costs and risks for patients. In this article, we propose a kernel assisted learning method for estimating the optimal individualized dose rule. The proposed methodology can also be applied to all other continuous decision-making problems. Advantages of the proposed method include robustness to model misspecification and capability of providing statistical inference for the estimated parameters. In the simulation studies, we show that this method is capable of identifying the optimal individualized dose rule and produces favorable expected outcomes in the population. Finally, we illustrate our approach using data from a warfarin dosing study for thrombosis patients.   
\end{abstract}
%%
%% The majority of ACM publications use numbered citations and
%% references.  The command \citestyle{authoryear} switches to the
%% If you are preparing content for an event
%% sponsored by ACM SIGGRAPH, you must use the "author year" style of
%% the next command will enable that style.

%\begin{CCSXML}
%<ccs2012>
%<concept>
%<concept_id>10010147.10010341.10010342</concept_id>
%<concept_desc>Computing methodologies~Model development and analysis</concept_desc>
%<concept_significance>500</concept_significance>
%</concept>
%<concept>
%<concept_id>10010147.10010257.10010293</concept_id>
%<concept_desc>Computing methodologies~Machine learning approaches</concept_desc>
%<concept_significance>300</concept_significance>
%</concept>
%<concept>
%<concept_id>10002950.10003648.10003662</concept_id>
%<concept_desc>Mathematics of computing~Probabilistic inference problems</concept_desc>
%<concept_significance>100</concept_significance>
%</concept>
%</ccs2012>
%\end{CCSXML}

%\ccsdesc[500]{Computing methodologies~Model development and analysis}
%\ccsdesc[300]{Computing methodologies~Machine learning approaches}
%\ccsdesc[100]{Mathematics of computing~Probabilistic inference problems}
%%




% seats. Ichiro Suzuki preparing to bat.}

\maketitle

%\begin{table*}
%  \caption{Some Typical Commands}
%  \label{tab:commands}
%  \begin{tabular}{ccl}
%    \toprule
%    Command &A Number & Comments\\
%    \midrule
%   \texttt{{\char'134}author} & 100& Author \\
%    \texttt{{\char'134}table}& 300 & For tables\\
%    \texttt{{\char'134}table*}& 400& For wider tables\\
%    \bottomrule
%  \end{tabular}
%\end{table*}

% an unnumbered equation:
%\begin{displaymath}
%  \sum_{i=0}^{\infty} x + 1
%\end{displaymath}

%\begin{figure}[h]
%%  \centering
%  \includegraphics[width=\linewidth]{sample-franklin}
%  \caption{1907 Franklin Model D roadster. Photograph by Harris \&
%%    Ewing, Inc. [Public domain], via Wikimedia
%    Commons. (\url{https://goo.gl/VLCRBB}).}
%  \Description{The 1907 Franklin Model D roadster.}
%\end{figure}

 %Figure captions go below the figure. 
 %Your figures should  also include a description suitable for screen readers, to
%assist the visually-challenged to better understand your work.

%If your work needs an appendix, add it before the
%``\verb|\end{document}|'' command at the conclusion of your source
%document.

%Start the appendix with the ``\verb|appendix|'' command:
%\begin{verbatim}
% \appendix
%\end{verbatim}
%and note that in the appendix, sections are lettered, not
%numbered. This document has two appendices, demonstrating the section
%and subsection identification method.


%------------------------------------
\section{Introduction}
%------------------------------------
\label{s:intro}
Statistical methods are increasingly popular for optimizing drug doses in clinical trials. A typical dose-finding study is conducted by a double-blind randomized trial where each patient is randomly assigned a dose among a few safe dose levels for a candidate drug \citep{chevret2006statistical}. At the end of the trial, a single dose leading to the best average treatment effect is determined as a recommendation for future patients. However, different patients might respond differently to the same dose of a drug due to their differences in physical conditions, genetic factors, environmental factors and medication histories. Taking these differences into consideration when making dose decisions is essential for achieving better treatment results. 

Recently, there has been a growing interest in personalized treatments optimization. However, most of the methods are restricted to a finite number of treatment options. In particular, people are interested in finding individualized treatment rules, which output a treatment option within a finite number of available treatments based on patient level information. Such treatment rules can thus be used to guide treatment decisions aiming to maximize the expected clinical outcome of interest, also known as the expected reward or value. An optimal treatment rule is defined to be the one that maximizes the value function in the population among a class of treatment rules. Various statistical learning methods have been proposed to infer optimal individualized treatment rules using data from randomized trials or observational studies. Existing methods include model-based approaches, such as the Q-learning  \citep{watkins1992q,zhao2009reinforcement,qian2011performance} and A-learning \citep{murphy2003optimal,robins2004optimal,henderson2010regret, liang2018deep, shi2018high}, direct value search methods by maximizing a nonparametric estimator of the value function \citep{zhang2012estimating,zhang2012robust,zhao2012estimating, fan2017concordance, shi2019concordance, zhou2019restricted}, and other semi-parametric methods \citep{song2017semiparametric, kang2018estimation, xiao2019robust}. 

The above methods, however, are not directly applicable when the number of possible treatment levels is large. 
Let us illustrate with warfarin, which is an anticoagulant drug commonly used for the prevention of thrombosis and thromboembolism. Establishing the appropriate dosage for warfarin is known to be a difficult problem because the optimal dosage can vary by a factor of 10 among patients, from 10mg to 100mg per week \citep{international2009estimation}. Incorrect dosages contribute largely to the adverse effects of warfarin usage. Underdosing will fail to alleviate symptoms in patients and overdosing will lead to catastrophic bleeding. In this case, an individualized dose rule, where a dose level is suggested within a continuous safe dose range according to each individual's physical conditions, would be better at tailoring to patient heterogeneity in drug response.  Several methods have been proposed for finding optimal individualized dose rules. One way of extending existing methods to the continuous dose case is to discretize the dose levels. \citet{laber2015tree} proposed a tree-based method and turned the problem into a classification problem by dividing patients into subgroups and assigning a single dose to each subgroup. \citet{chen2018estimating} extended the outcome weighted learning method \citep{zhao2012estimating} from binary treatment settings to ordinal treatment settings. However, in cases where patient responses are sensitive to dose changes, a discretized dose rule with a small number of levels will fail to provide dose recommendations leading to optimal clinical results. On the other hand, a discretized dose rule with a large number of levels may result in limited observations within each subgroup, and thus may be at risk of overfitting.  

Alternatively, \citet{rich2014simulating} extended the Q-learning method by modeling the interactions between the dose level and covariates with both linear and quadratic terms in doses. However, such a parametric approach is sensitive to model misspecification and the estimated individualized dose rule might be far away from the true optimal dose rule. In addition, it cannot be guaranteed that the estimated optimal dose falls into the safe dose range. More recently,  \citet{chen2016personalized} extended the outcome weighted learning method proposed by \citet{zhao2012estimating} and transformed the dose-finding problem into a weighted regression with individual rewards as weights. The optimal dose rule is then obtained by optimizing a non-convex loss function. This method is robust to model misspecification and has appealing computational properties, however, the associated statistical inference for the estimated dose rule is challenging to determine. 
In this article we propose a kernel assisted learning method to infer the optimal individualized dose rule in a manner which enables statistical inference. Our proposed method can be viewed as a direct value search method. Specifically, we first estimate the value function with a kernel based estimator. Then we search for the optimal individualized dose rule within a prespecified class of rules where the suggested doses always lie in the safe dose range. The proposed method is robust to model misspecification and is applicable to data from both randomized trials and observational studies. We establish the consistency and asymptotic normality of the estimated parameters in the obtained optimal dose rule. In particular, the asymptotic covariance of the estimators is derived based on nontrivial calculations of the expectation of a U-statistic. 

The remainder of the article is organized as follows. In Section 2, we present the problem setting and our proposed method. The theoretical results of the estimated parameters are established in Section 3. In Section 4, we demonstrate the empirical performance of the proposed method via simulations. In Section 5, the proposed method is further illustrated with an application to a warfarin study. Some discussions and conclusions are given in Section 6. Proofs of the theoretical results are provided in the appendix.

%-------------------------------------------------------
\section{Method}
%-------------------------------------------------------

\subsection{Problem Setting}
The observed data consist of $n$ independent and identically distributed observations $\{(X_i,$ $A_i,$ $Y_i)\}_{i=1}^n$, where $X_i\in\mathcal X$ is a $d$ dimensional vector of covariates for the $i$th patient, $A_i\in\mathcal A$ is the dose assigned to the patient with $\mathcal A$ being the safe dose range, and $Y_i\in\mathbb R$ is the outcome of interest. Without loss of generality, we assume that larger $Y$ means better outcome. Let $\pi({X})$ denote an individualized dose rule, which is a deterministic mapping function from $\mathcal{X}$ to $\mathcal{A}$. To define the value function of an individualized dose rule, we use the potential outcome framework \citep{rubin1978bayesian}. Specifically, let $Y^*(a)$ be the potential outcome that would be observed when a dose level $a\in \mathcal A$ is given. Define the value function as the expected potential outcome in the population if everyone follows the dose rule $\pi$, i.e. $V(\pi)=E[Y^*\{\pi({X})\}]$. The optimal individualized dose rule is defined as $\pi^{opt}= \arg\max V(\pi)$.

In order to estimate the value function from the observed data, we need to make the following three assumptions similar to those adopted in the causal inference literature \citep{robins2004optimal}. First, we assume $Y=\int_{\mathcal A} \delta(A=a)Y^*(a)da$, where $\delta(\cdot)$ is the Dirac delta function. This corresponds to the stable unit treatment value assumption (also known as the consistency assumption). It assumes that the observed outcome is the same as the potential outcome had the dosage given to the patient be the actual dose. This assumption also implies that there is no interference among patients. 
Second, we assume that the potential outcomes $\{Y^*(a):a\in\mathcal A\}$ are conditionally independent of $A$ given ${X}$, which is also known as the no unmeasured confounders assumption. %This assumption can be naturally satisfied by the design of a randomized dose trial. However, it cannot be validated in an observational study. 
Third, we assume that there exists a $c>0$ such that $p(A=a|{X}={x})\ge c$ for all $a\in \mathcal A, {x}\in \mathcal{X}$, where $p(a|{x})$ is the conditional density of $A=a$ given ${X}={x}$. This is a generalization of the positivity assumption for continuous dosing. Under these assumptions, we can show that $V(\pi)$ can be estimated with the observed data:
 \begin{align*}
    V(\pi)&=E[Y^*(\pi( X))]
    =E_{ X}[E\{Y^*(\pi( X))| X\}]\\
    &=E_{ X}[E\{Y^*(\pi( X))|A=\pi( X), X\}]
    =E_ X[E\{Y|A=\pi( X), X\}].
\end{align*} 
The second equation above is based on the basic property of conditional densities. The third equation above is valid because of the no unmeasured confounder assumption. The fourth equation is based on the consistency assumption. The positivity assumption ensures that the right side of the last equation can be estimated empirically.  In the next section, we will propose a consistent estimator for $V(\pi)$ based on kernel smoothing. 

\subsection{Method}
To estimate the optimal IDR, we first estimate $V(\pi)$ with a kernel based estimator $\hat V(\pi)$ and then estimate $\pi^{opt}$ by directly maximizing the estimated value function $\hat V(\pi)$. We search for the optimal individualized dose rule within a class of dose rules of the form: $\pi_{{\beta}}({x})=\pi({x};{{\beta}})\in\mathcal G$, where $\mathcal G=\{g({\beta}^T{x}),{\beta}\in \mathbb R^d\}$, and $g:\mathbb R\to \mathcal A$ is a predefined link function to ensure that the suggested dosage is within the safe dose range. Thus $\hat \pi^{opt}=\arg\max_{\pi\in\mathcal G} \hat V(\pi).$

Notice that $\hat\pi^{opt}$ is an estimator of the optimal IDR within $\mathcal G$:
$
   \pi_{{\beta}}^{opt}= \arg\max_{\pi\in\mathcal G} V(\pi)=\pi_{\beta^*},
$ where ${\beta}^*=\arg\max_{{\beta}}V(\pi_{{\beta}})$.
If the true optimal IDR lies in $\mathcal G$, then the proposed $\pi^{opt}_{\beta}(X)=\pi^{opt}$. To see this more clearly, we illustrate with a toy example. If the true model for $E(Y|A,X)$ takes the form: $E(Y| A, X)=\tilde \mu({X}) + Q\{A-g({\tilde\beta}^T {X})\} H({X})$, where $\tilde\mu({X})$ is an unspecified baseline function, $H({X})$ is a non-negative function and $Q(\cdot)$ is a unimodal function which is maximized at 0, then $E(Y \mid A, X)$ is maximized at dose level $A=g({\tilde \beta}^T{x})$ for patients with covariates ${X}={x}$. Thus, the true optimal individualized dose rule is:
\begin{align*}
        \pi^{opt}(X)&=\arg\max_{\pi} V(\pi)\\
        &=\arg\max_{\pi}E_X\Big[ E\big\{Y \mid A=\pi({X}),{X}\big\}\Big]\\
        &=\arg\max_{\pi}E_X\Big[\tilde \mu({X}) + Q\big\{\pi({X})-g(\tilde{\beta}^T {X})\big\} H({X})\Big]\\
        &=\arg\max_{\pi}E_X\Big[Q\big\{\pi({X})-g(\tilde{\beta}^T {X})\big\} H({X})\Big]\\
        &=g(\tilde{\beta}^T {X})\in \mathcal G.
\end{align*}
The last equation above is true because $Q\{\pi(X)-g(\tilde\beta^TX)\}H(X)$ is maximized at $g(\tilde \beta ^T X)$ for each $X\in \mathcal X$.
If a unique maximizer of $V(\pi_{{\beta}})$ exists, then 
        \begin{align*}
            &{\beta}^*=\arg\max_{{\beta}} V(\pi_{{\beta}})\\
        &=\arg\max_{{\beta}}E_X\Big[Q\big\{g({\beta}^T {X})-g({\tilde\beta}^T {X})\big\} H({X})\Big]
        ={\tilde \beta}.
        \end{align*}
Therefore, $\pi_{{\beta}}^{opt}=g({{\beta}^*}^T{X})=\pi^{opt}$. Notice that if, 
$\pi^{opt}=\notin \mathcal G,$ then $\pi_\beta^{opt}\neq \pi^{opt}$. However, $\pi^{opt}_\beta$ is still of interest as long as the form of $\mathcal G$ is flexible enough, because it maximizes the value function among this set of treatment rules.$\beta^*$ can be estimated using $\hat \beta=\arg\max \hat V(\pi_\beta)$, and the optimal IDR within $\mathcal G$ can be thus estimated with  $\hat\pi_{\beta}^{opt}=\pi(X;\hat\beta)$. Notice that we do not need any model assumption on the form of the conditional expectation $E(Y|A,X)$ to apply this method. 

Next, we propose a kernel based estimator for the value function. Let 
$$M({\beta})=V\big(\pi_{{\beta}}\big)=\int_{ x\in\mathcal X} m\{{x},g({\beta}^T{x})\}f( x)d x,$$
where
    $m({x},a)=E(Y\mid {X}={x},A=a)$ and $f({x})$ is the marginal density of ${X}$. Thus, ${\beta}^*=\arg\max_{{\beta}} M({\beta})$. The function $m\{{x},g({\beta}^T{x})\}$ is estimated using the Nadaraya-Watson method given:
    $$
    \hat m\{x,g(\beta^T x)\}= 
    \frac{\frac{1}{n} \sum_{i=1}^n Y_i \frac{1}{h_x^d} K_d (\frac{ x- X_i}{h_x})\frac{1}{h_a} K\big\{\frac{g({\beta}^T{x})-A_i}{h_a}\big\}}
            {\frac{1}{n} \sum_{i=1}^n  \frac{1}{h_x^d}{K_d}(\frac{ x- X_i}{h_x})\frac{1}{h_a} K\big\{\frac{g({\beta}^T{x})-A_i}{h_a}\big\}},
    $$
    where $K(\cdot)$ is a univariate kernel function and $K_d(\cdot)$ is a $d$ dimensional kernel function. Here, $h_x$ and $h_a$ are bandwidths that go to 0 as $n\to \infty$. Note that for simplicity of notification, we use the same bandwidth for all dimensions of ${X}$ here. In practice, we can use different bandwidths for different dimensions of ${X}$ to increase the efficiency of the estimation. Moreover, the marginal density of ${X}$ is estimated by $\hat f({x})= (1/n) \sum_{i=1}^n  {K_d}\{( x- X_i)/{h_x}\}/h_x^d$. The estimated value function can thus be written as:
\begin{align*}
            M_n({\beta})=\int_{ x} 
            & \Big[\frac{\frac{1}{n} \sum_{i=1}^n Y_i \frac{1}{h_x^d} K_d (\frac{ x- X_i}{h_x})\frac{1}{h_a} K\big\{\frac{g({\beta}^T{x})-A_i}{h_a}\big\}}
            {\frac{1}{n} \sum_{i=1}^n  \frac{1}{h_x^d}{K_d}(\frac{ x- X_i}{h_x})\frac{1}{h_a} K\big\{\frac{g({\beta}^T{x})-A_i}{h_a}\big\}}\Big]
            \Big[\frac{1}{n}\sum_{i=1}^n \frac{1}{h_x^d}{K_d}(\frac{ x- X_i}{h_x})\Big]d x.
\end{align*}
Then ${\beta}^*$ is estimated with $\hat{{\beta}}_n=\arg\max_{{\beta}\in\Theta} M_n\big({\beta}\big)$, where $\Theta$ is a compact subset of $\mathbb R^d$ containing ${\beta}^*$.
%----------------------------------------------
\subsection{Computational Details}

To implement the proposed method, the R package optimr() is used for optimization of the objective function. The integral in $M_n({\beta})$ is estimated by taking the average of $q$ grid points in the covariate space. In our implementation, we chose $q = 3000$. In order to find the global maximizer of $M_n({\beta})$, we start optimization from $d$ different initial points $\{(1,0,...,0), (0,1,0,...,0),$ $...,(0,...,0,1)\}$ and choose the one that leads to the maximal objective function value. Denote the maximizer as $\hat{{\beta}}_n$. When there is only one continuous covariate included, following the theoretical rate of the bandwidth parameters, the bandwidths are chosen as $h_x=C_x \text{sd}(X)n^{-1/4.5}$, $h_a=C_a \text{sd}(A)n^{-1/4.5}$, where $C_x$ and $C_a$ are constants between $0.5$ and $3.5$. %When there are multiple continuous covariates, different bandwidths are used for each dimension of ${X}$. Choosing the constants in the bandwidths is thus nontrivial. For the simulation setting 5, we arbitrarily choose the rate of the bandwidths as $n^{-1/8}$ and then use 5-fold cross validation to choose the constants $C_x$, $C_a$ which minimize the mean squared error. % of the Nadaraya-Watson estimator.

When the covariates consist of both continuous variables and categorical variables, the categorical variables are stratified for estimation of the value function. Specifically, assume that ${X}=({X}_{1}^T,{X}_{2}^T)^T\in\mathcal X$, where ${X}_{1}\in \mathcal C$ is a $d_1$ dimensional vector of continuous variables and ${X}_{2}\in \mathcal D$ is a $d_2$ dimensional vector of categorical variables. The form of $M_n({\beta})$ then becomes:
\begin{align*}
            M_n({\beta})=\sum_{ x_2\in \mathcal D} \int_{ x_1\in \mathcal C} & \Big[\frac{\frac{1}{n} \sum_{i=1}^n Y_i \tilde K( x,  X_i)\frac{1}{h_a} K\big\{\frac{g({\beta}^T{X})-A_i}{h_a}\big\}}
            {\frac{1}{n} \sum_{i=1}^n  \tilde K( x,  X_i)\frac{1}{h_a} K\big\{\frac{g({\beta}^T{X})-A_i}{h_a}\big\}}\Big]
            \Big\{\frac{1}{n}\sum_{i=1}^n \tilde K( x,  X_i)\Big\}d x_1,
\end{align*}
        where $ x=( x_1^T,  x_2^T)^T$, $ X_i=( X_{i1}^T,  X_{i2}^T)^T$, \\$\tilde K( x,  X_i)=(1/h_x^{d_1}) K_{d_1} \{( x_1- X_{i1})/{h_x}\}I( X_{i2}=x_2)$.
        
The R code for the proposed method is available at: \url{https://github.com/lz2379/Kernel_Assisted_Learning}

%-------------------------------------------------------
\section{Theoretical Results}\label{theoretical-section}
%-------------------------------------------------------

In this section, we establish the asymptotic  properties of $\hat{{\beta}}_n$. To prove these results, we need to make the following assumptions. In the following equations, $\dot f(x)$, $\ddot f(x)$ and $\dddot f(x)$ denote the first, second and third derivatives of the function $f$ with respect to $x$; $\kappa_{0,2}=\int_\mathbb R K^2(v)dv$ and $\dot\kappa_{0,2}=\int_\mathbb R \dot K^2(v)dv$.

    \begin{assumption}{1}
    \label{A1}
    Assume that $\dot K(v)$, $\ddot K(v)$ and $\dddot K(v)$ exist for $v\in \mathcal V$, where $\mathcal V$ is a subset of $\mathbb R$. For $h\to 0^+$ as $n\to \infty$ and constants $l$,$u$ such that $l<0< u$, 
$\int_{[l/h,u/h]\cap \mathcal V}K(v)dv$ $=1-O(h^2)$,
$\int_{[l/h,u/h]\cap \mathcal V}vK(v)dv$ $=O(h)$,
$\int_{[l/h,u/h]\cap \mathcal V}\dot K(v)dv$ $=O(h^3)$,
$\int_{[l/h,u/h]\cap \mathcal V}-v\dot K(v)dv$ $= 1-O(h^2)$,
$\int_{[l/h,u/h]\cap \mathcal V}K^2(v)dv$ $=\kappa_{0,2}-O(h^2)$,
$\int_{[l/h,u/h]\cap \mathcal V}\dot K^2(v)dv$ $=\dot\kappa_{0,2}-O(h^2)$,
$\int_{[l/h,u/h]\cap \mathcal V}\ddot K(v)dv$ $=O(h^4)$,
$\int_{[l/h,u/h]\cap \mathcal V}v\ddot K(v)dv$ $=O(h^3)$,
$\int_{[l/h,u/h]\cap \mathcal V}v^2\ddot K(v)dv$ $=2-O(h^2)$,
$\int_{[l/h,u/h]\cap \mathcal V}\dddot K(v)dv$ $=O(h^3)$,
$\int_{[l/h,u/h]\cap \mathcal V}v \dddot K(v)dv$ $=O(h^2)$,
$\int_{[l/h,u/h]\cap \mathcal V}v^2 \dddot K(v)dv$ $=O(h)$,
$\int_{[l/h,u/h]\cap \mathcal V}v^3 \dddot K(v)dv$ $=O(1)$.
    \end{assumption}
    
    \begin{assumption}
    \label{A2} The function $M({\beta})=V(\pi_{\beta})$ has a unique maximizer ${\beta}^*$. 
    \end{assumption}
    
    \begin{assumption}
    \label{A3} The function $m({x},a)$ is uniformly bounded. The joint density function of ${X}$ and $A$, $f({x},a)$, is uniformly bounded away from 0. In addition, the first, second, third and fourth order derivatives of $m(x,a)$ and $f({x},a)$ with respect to $x$ and $a$ exist and are uniformly bounded almost everywhere. %$f({x},a)$ is bounded away from 0. 
    \end{assumption}
    
    \begin{assumption}
    \label{A4} The covariate $X$ has bounded first, second and third moments. %for the existence of covariance
    \end{assumption}
    
    \begin{assumption} 
    \label{A5} The function $g(\cdot)$ is thrice differentiable almost everywhere and the corresponding derivatives, 
    $\dot g((\cdot),$ $\ddot g(\cdot),$ $\dddot g(\cdot)$ are bounded almost everywhere.
    \end{assumption}
Assumption \ref{A1} can be satisfied by most commonly used kernel functions such as the Gaussian kernel function and all sufficiently smooth bounded kernel functions. Assumption \ref{A2} is an identifiability condition for ${\beta}^*$. Assumptions \ref{A3}--\ref{A5} ensure the existence of the limit of the expectation of $M_n({\beta})$ and the existence of the covariance matrix of the limiting distribution. In the following two theorems, we establish the consistency and asymptotic normality of $\hat{{\beta}}_n$, respectively.
% Discuss the conditions assumed.
\begin{theorem}\label{theorem1}
    Under assumptions \ref{A1}--\ref{A3}, for $h_x$, $h_a$ satisfying \\$nh_x^{2d}h_a^2$ $\to \infty$ as $n\to \infty$, we have
$
    \sup_{{\beta}\in \Theta}|M_n({\beta})- M({\beta})|
$ converge in probability to $0$, 
where $\Theta$ is a compact region containing ${\beta}^*$. Thus, $\hat{{\beta}}_n=\arg\max_{{\beta}\in\Theta} M_n({\beta})$ converge in probability to ${\beta}^*$. 
\end{theorem}
%Theorem $\ref{theorem1}$ implies that the estimated $\pi_{{\beta}}^{opt}$ converges to $\pi^{opt}$ if the prespecified class of dose rules $\mathcal G$ contains $\pi^{opt}$. Next, we derive the convergence rate and the limiting distribution of $\hat{{\beta}}_n$.

\begin{theorem}\label{theorem2}
     Under assumptions \ref{A1}--\ref{A5}, for $h_x$, $h_a$ satisfying \\$nh_x^{2d}h_a^2$ $\to \infty$ and $nh_x^dh_a^3\to \infty$ as $n\to \infty$, we have
     \begin{gather*}
         (nh_x^dh_a^3)^{1/2}(\hat{{\beta}}_n-{\beta}^*) \to N\big\{0,{ D} ({\beta}^*)^{-1}{\Sigma}_{ S}({\beta}^*){ D}^{-1}({\beta}^*)\big\} 
     \end{gather*}
     in distribution as $n\to\infty$, where
     \begin{align*}
          D({\beta}^*) =\int_x &\Big[ {m}_{aa} \big\{{x},g({\beta}^T{x})\big\}\dot g^2({\beta}^T{x}) +
          {m}_a\big\{{x},g({\beta}^T{x})\big\}\ddot g({\beta}^T{x})\Big]f({x}) {x x}^T d{x},
    \end{align*}
    \begin{align*}
         {\Sigma}_{ S}({\beta}^*)=\int_{{x}} &\Big[\dot g^2({\beta}^T{x}){xx}^T {\kappa}_{0,2}\dot \kappa_{0,2} f^2({x})\Big]
         \Big[\frac{m_2\big\{{x},g({x}^T{\beta})\big\}-m^2\big\{{x},g({x}^T{\beta})\big\}}{f\big\{{x},g({x}^T{\beta})\big\}}\Big]d{x},
     \end{align*}
 and  $m_a( x,a)=\partial m( x,a)/{\partial a}$,
    $m_{aa}( x,a)= \partial^2 m( x,a)/{\partial a ^2}$, $m_2( x,a)=E(Y^2\mid {X}={x},A=a)$.
\end{theorem}
Proofs of the above theorems are based on theory for kernel density estimators \citep{schuster1969estimation} and M-estimation \citep{kosorok2008introduction}. Details of proofs are given in the appendix. Note that the convergence rate is slower then $n^{1/2}$ due to the kernel estimation of the value function. %We should comment on if there are similar type of theoretical results and how to do the variance estimation here. or some other remarks you can think of???)

%-------------------------------------------------------
\section{Simulation Studies}
%-------------------------------------------------------

In this section, we conduct some simulations to show the capability of our proposed method in identifying the optimal individualized dose rule. We first simulate some simple settings with only one covariate.  $X$ is generated randomly from the standard normal distribution. $A$ is generated from the uniform distribution on $[0,1]$. We generate $X$ and $A$ independently to mimic a randomized dose trial where a random dose from the safe dose range is assigned to each patient. The optimal dose rule is $\pi^{opt}(X)=g(\tilde \beta_0+ \tilde \beta_1X)$, where $g(s)=1/\{1+\text{exp}(-s)\}$. $Y$ is generated from a normal distribution with mean $m(A,X)$ and standard deviation $0.5$, where $m(A,X)=\tilde \mu(X)-10\{A-\pi^{opt}(X)\}^2$. We use two different baseline functions for $\tilde \mu(X)$ and two different sets of $(\tilde \beta_0,\tilde \beta_1)$ as shown in Table \ref{t1:sim}. The sample sizes are $n=400$ and $n=800$ and each setting is replicated 500 times. 

The average bias and the standard deviation of the estimated parameters from 500 simulations are summarized in the first half of Table \ref{table:2}. The estimated parameters were close to the true parameters. The third column shows the average of the standard errors estimated with the covariance function formula derived in Theorem \ref{theorem2} (see appendix for details). $95\%$ confidence intervals were calculated with the estimated standard errors. The coverage probabilities are shown in the table. From the result, we can see that the bias and standard deviation of the estimated parameters decreased with larger sample sizes. The coverage probabilities of the confidence intervals were close to $95\%$, supporting the convergence results given in Section \ref{theoretical-section}. 
 
We also study the performance of our method when the training data are from observational studies, where the doses given to the patients may depend on the covariates $X$. The simulation settings are the same as settings 1--4 except that $A$ is generated from the distribution $\text{beta}\big\{2\exp(\tilde \beta_0+\tilde \beta_1X),2\big\}$. The results are summarized in the second half of Table \ref{table:2}. The proposed method was still capable of giving good estimates of the parameters and the coverage of the confidence intervals were close to $95\%$. These simulation implies that the proposed method performs well with data from both randomized trials and observational studies.

 \begin{table}
    \caption{Summary of simulation settings\label{t1:sim}}
    \begin{center}
    \begin{tabular}{c c c}
        \toprule
        &No baseline& With baseline  \\
        &$\tilde \mu(X)=0$&$\tilde \mu(X)=1+0.5 \cos(2\pi X)$\\
\cmidrule(lr){2-3}
        $\tilde \beta_0=0, \tilde \beta_1=0.5$& Setting 1&Setting 3\\
        $\tilde \beta_0=0, \tilde \beta_1=1$& Setting 2& Setting 4\\
        \bottomrule
    \end{tabular}
    \end{center}  
\end{table}

\begin{table}
\setlength{\tabcolsep}{4pt}
\begin{center}
 \caption{Simulation results from 500 replicates for randomized trials and observational studies.\label{table:2}}
\begin{threeparttable}
 \begin{tabular}{c c c c c c c c c c}
\toprule
 \multicolumn{10}{c}{Randomized trials}\\
\midrule
    \multicolumn{2}{c}{}& \multicolumn{4}{c}{$\hat \beta_{n,0}$}& \multicolumn{4}{c}{$\hat \beta_{n,1}$}\\
\cmidrule(lr){3-6} \cmidrule(lr){7-10}
   & n   & Bias$^*$   & SD$^*$  & SE$^*$ & CP & Bias$^*$  & SD$^*$    & SE$^*$ & CP \\
 \cmidrule(lr){1-2} \cmidrule(lr){3-6} \cmidrule(lr){7-10}

 S1       & 400 & 2.5 &46.6 & 47.5& 95.6     & -17.3 & 53.5 & 54.5  & 92.8     \\
         & 800 & 2.4 &33.7 & 33.4& 95.8     & -19.5 & 37.3 & 38.5 & 90.2    \\
 S2       & 400 & 2.1 &52.2 & 54.4 & 95.6     & 38.9 & 91.0 & 93.7 & 94.6    \\
         & 800 & 1.5 &39.1 & 38.1 & 93.8     & 33.0 & 63.0 & 65.9 & 95.8    \\
 S3       & 400 & 2.7 &54.1  & 55.7 & 95.2     & -20.4 & 64.5 &  64.1 & 90.8    \\
         & 800 & 1.6 &38.8  & 39.3 & 95.0     & -18.9 & 43.7 &  45.4 & 92.0   \\
 S4       & 400 & 2.4 &61.8  & 63.4 & 95.4     & 39.4 & 103.5 & 111.2 & 96.2    \\
         & 800 & -1.5 &44.4  & 44.3 & 94.6     & 33.6 & 75.0 & 77.5 & 95.6    \\
\midrule
\multicolumn{10}{c}{Observational studies}\\
\midrule
\multicolumn{2}{c}{}& \multicolumn{4}{c}{$\hat \beta_{n,0}$}& \multicolumn{4}{c}{$\hat \beta_{n,1}$}\\
\cmidrule(lr){3-6} \cmidrule(lr){7-10}
   & n   & Bias$^*$   & SD$^*$  & SE$^*$ & CP & Bias$^*$  & SD$^*$    & SE$^*$ & CP \\
 \cmidrule(lr){1-2}  \cmidrule(lr){3-6} \cmidrule(lr){7-10}
 S1       & 400 & 13.9 & 80.5 & 82.4 & 96.0     & 32.4 & 97.7 & 102.1  & 94.6     \\
         & 800 & 8.5   & 47.3 & 47.0 & 94.6     & -19.6 &56.7 & 58.2 & 92.2    \\
 S2       & 400 & 21.9 &83.3 & 88.1 & 96.4     & 7.6& 146.4 & 150.4 & 95.2    \\
         & 800 & 17.8  &63.4 & 60.9 & 93.0     & 0.6 & 94.0 & 103.2 & 98.2    \\
 S3       & 400 & 13.4 & 89.4 & 90.2 & 95.6     & 33.2 & 109.3 & 112.3 & 94.8     \\
         & 800 &  9.0 & 53.1  & 50.6 & 93.0        & -22.8 & 60.2 &  63.2 & 93.4    \\
 S4       & 400 & 21.6 &91.3  & 97.1 & 96.2     & 5.0 & 165.4 & 169.3 & 95.8    \\
         & 800 & 20.3 &71.2  & 67.6 & 93.2     & 2.0 & 109.0 & 116.8 & 97.0    \\
\bottomrule
\end{tabular}
\begin{tablenotes}
\footnotesize{
\item[1] Note: * columns are in $10^{-3}$ scale
\item[2] Note: SD refers to the standard deviation of the estimated parameters from 500 replicates, SE refers to the mean of the estimated standard errors calculated by our covariance function, CP refers to the coverage probability of the $95\%$ confidence intervals calculated using the estimated standard errors.}
\item[3] Note: The worst case Monte Carlo standard error for proportions is $1.3\%$.
\end{tablenotes}
\end{threeparttable}
\end{center}
\end{table}

Under settings 1--4, we compare our method with linear based O-learning (LOL) and kernel based O-learning (KOL) proposed in \citet{chen2016personalized} and a discretized dose rule estimated using Q-learning.  For discretized Q-learning, we divide the safe dose range into 10 equally spaced intervals: $\mathcal A=\mathcal A_1\cup \dots \cup \mathcal A_{10}$ and create an indicator variable for each of the dose intervals ${I}=(I_1,I_2,\dots,I_{10})$, where $I_j=I(A\in \mathcal A_j)$, $j=1,\dots, 10$. The covariates included in the regression models are $({X}, {X}^2, {I}, {IX},{IX^2} )$. To this end, an optimal dose range is selected for each individual and the middle point of the selected interval is suggested to the patient. The results from $500$ replicates are summarized in Table \ref{table:setting1-4value}.  Each column is the average value function of the dose rule estimated by the corresponding method. The value function is evaluated at a testing dataset. The numbers in the parentheses are the standard deviation of the estimated value functions. From Table \ref{table:setting1-4value}, we see that the proposed method performed the best under most settings. In the simulation for observational studies, O-learning performed the best when the sample size was small. However, the proposed method performed comparatively well and performed better as the sample size increased. The discretized Q-learning method did not provide a good dose suggestion in this case.
%--------------------------------------------
\begin{table}
\begin{center}
\caption{Value estimate $ V(\hat\pi)$ from 500 simulations in settings 1-4\label{table:setting1-4value}} 
\begin{threeparttable}
 \begin{tabular}{c c c c c c}
\toprule
  \multicolumn{6}{c}{Randomized trials} \\
\midrule
         & n &DQ& LOL & KOL & KAL  \\
\cmidrule(lr){1-2}  \cmidrule(lr){3-6}
    S1&  400& -38.1(1.9)& -7.7(7.0)& -16.5(8.2)&\textbf{-2.7(2.7)}\\
      &  800& -32.7(0.9)& -3.9(3.7)& -9.3(4.3)&\textbf{-1.9(1.8)}\\

    S2&  400& -33.9(1.3)& -18.1(10.5)& -31.7(12.5)&\textbf{-3.3(3.7)}\\
      &  800& -20.0(0.8)& -15.6(7.5)& -20.4(7.4)&\textbf{-1.9(1.8)}\\

    S3&  400& -41.6(11.4)& -8.5(14.1)& -17.2(14.2)&\textbf{-3.7(12.4)}\\
      &  800& -61.2(11.5)& -4.3(12.1)& -10.0(12.7)&\textbf{-2.4(11.7)}\\

    S4&  400& -52.5(11.9)& -21.3(17.7)& -33.3(17.4)&\textbf{-4.2(12.5)}\\
      &  800& -23.3(11.9)& -17.8(14.5)& -22.4(13.9)&\textbf{-2.4(11.7)}\\
         
\midrule
\multicolumn{6}{c}{Observational studies}\\
\midrule
         & n &DQ& LOL& KOL & KAL  \\
\cmidrule(lr){1-2}  \cmidrule(lr){3-6}
    S1&  400&  -29.5(1.2)&\textbf{-7.4(6.3)}& -15.6(7.7) &-8.1(8.0)\\
      &  800&  -24.4(0.9)& -5.5(4.3)& -10.3(5.3)&\textbf{-3.1(3.1)}\\

    S2&  400&  -16.0(7.6)&-14.1(6.6) &  -21.3(9.7)&\textbf{-8.2(8.3)}\\
      &  800&  -32.0(1.1)& -12.8(4.7)& -12.2(4.5)&\textbf{-4.4(4.4)}\\

    S3&  400&  -29.1(11.6)& \textbf{-8.1(13.5)}& -11.7(14.6)&-9.8(15.7)\\
      &  800&  -34.2(11.8)& -6.2(12.4)& -11.2(13.3)&\textbf{-3.5(12.0)}\\

    S4&  400&  -83.8(13.8)& -14.7(13.0)& -20.7(15.8)&\textbf{-10.0(15.8)}\\
      &  800&  -34.1(11.3)& -13.5(12.1)& -11.2(12.3)&\textbf{-5.1(12.2)}\\
\bottomrule
\end{tabular}
\begin{tablenotes}
\footnotesize{
\item[1] Note: DQ refers to discretized Q-learning, LOL refers to linear O-learning, KOL refers to kernel based O-learning, KAL refers to kernel assisted learning.
\item[2] All columns are in $10^{-3}$ scale. For settings 3 and 4, the numbers in the table are the value estimate $V(\hat\pi)-1$ for the purpose of comparison with the first two settings.}
\end{tablenotes}
\end{threeparttable}
\end{center}
\end{table}

%-------------------------------------------------------
\section{Warfarin Data Analysis}
%-------------------------------------------------------

Warfarin is a widely used anticoagulant for prevention of thrombosis and thromboembolism. Although highly efficacious, dosing for warfarin is known to be challenging because of the narrow therapeutic index and the large variability among patients \citep{johnson2011clinical}. Overdosing of warfarin leads to bleeding and underdosing diminishes the efficacy of the medication. The international normalized ratio (INR) measures the clotting tendency of the blood. An INR between $2$--$3$ is considered to be safe and efficacious for patients. Typically, the warfarin dosage is decided empirically: an initial dose is given based on the population average, and adjustments are made in the subsequent weeks while the INR of the patient is being tracked. A stable dose is decided in the end to achieve an INR of $2$--$3$ \citep{johnson2011clinical}. The dosing process may take weeks to months, during which the patient is constantly at risk of bleeding or under-dosing. Therefore, a quantitative method for warfarin dosing will greatly decrease the time, cost and risks for patients. 

The following analysis uses the warfarin dataset collected by \citet{international2009estimation}. In the original paper, a linear regression was used to predict the stable dose using clinical results and pharmacogenetic information, including age, weight, height, gender, race, two kinds of medications (Cytochrom P450 and Amiodarone), and two genotypes (CYP2C9 genotype and VKORC1 genotype). This prediction method is based on the assumption that the stable doses received by the patients are optimal. However, later studies showed that the suggested doses by the International Warfarin Pharmacogenetic Consortium are suboptimal for elderly people, implying that the optimal dose assumption might not be valid \citep{chen2016personalized}. 

We apply our proposed method to this dataset to estimate the optimal individualized dose rule for warfarin. Instead of using only the data of the patients with stablized INR, we include all patients who received weekly doses between 6 mg to 95 mg. The medication information is missing for half of the observations and is therefore excluded from our analysis. Observations which are missing in the other variables are removed from the dataset, resulting in a total of $3567$ patients. The outcome variable is defined as $Y_i=-(\text{INR}_i-2.5)^2$ for the $i$-th observation. Stratification of the categorical variables is needed for the kernel density estimation. In order to ensure that there are enough observations in each stratified group, we consider only categorical variables that are distributed comparatively even among different groups. In our analysis, we use three variables: height, gender and the indicator variable for VKORC1 of type AG. Before we apply the proposed the method, we normalize all the variables by $X_{i,j}=(X_{i,j}- \bar X_j)/\text{sd}(X_j)$, where $j=1,2,3$, $i=1,\dots, n$. $\bar X_j=\sum_{i=1}^n X_{i,j}/n$ and sd$(X_j)$ is the standard deviation of the $j$-th variable. %variable selection?

The estimation results are shown in Table \ref{t5}. The p-value is obtained for each of the parameters. The result implied that the optimal dose for male is higher than the optimal dose for female given all the other variables being the same. It was also implied that the patients with genotype VKORC1$=$AG need higher doses than the patients with VKORC1$\neq$ AG.  We use the same variables and compare our method with O-learning and the discretized Q-learning method. For the discretized Q-learning method, we also divide the dose range into 10 equally spaced intervals. The suggested doses by the three methods are shown in Fig. \ref{fig5:suggest}. The result shows a tendency of the discretized Q-learning to suggest extreme doses, which is not ideal in real application. This might be due to the fact that the higher dose intervals contain small numbers of observations, and thus the estimated models are dominated by a few subjects. 

To evaluate the estimated dose rules of these methods, we randomly take two thirds of the data as training data and the rest of the data as testing data. The optimal individualized dose rule is estimated with the training data. The value function of the suggested individualized dose rule $V(\hat\pi)$ is estimated with the average of the Nadaraya Watson estimator for $E\{Y\mid {X}, A=\hat\pi({X})\}$ in the testing dataset. The tuning parameters for the Nadaray-Watson estimators are taken as $h_x=1.25 \text{sd}(X) n_{\text{test}}^{-1/4.5}$ and $h_a=1.75 \text{sd}(A) n_{\text{test}}^{-1/4.5}$, where $n_{\text{test}}=1189$ is the size of the testing dataset. The process is repeated 200 times. The distribution of the estimated value of the suggested dose is shown in Fig. \ref{fig6:value}. The suggested individualized dose rule with our proposed method lead to better expected outcomes in the population compared to the other methods. The performance of the discretized Q-learning method was not stable as shown in the result. However, this result was only based on the three variables selected, while in reality, the two medications (Cytochrome P450 enzyme and Amiodarone) and the genotype CYP2C9 are also of significant importance in warfarin dosing.  The computation complexity of our proposed method restricted its capability of handling higher dimensional problems.  

\begin{figure}
    \centerline{\includegraphics[width=\columnwidth]{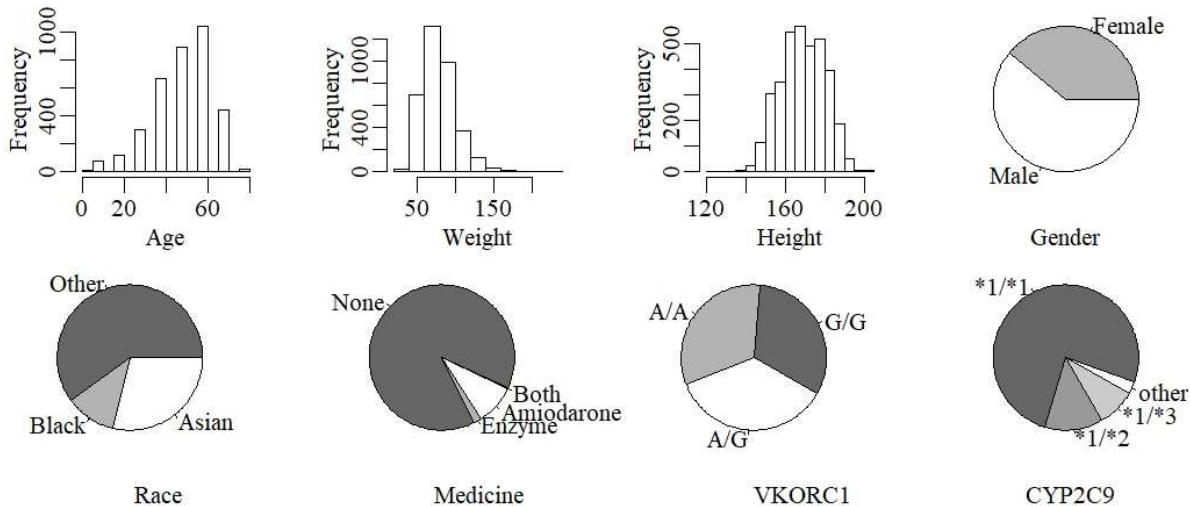}}
    \caption{Distribution of the variables in the warfarin dataset.\label{fig3:var}}
\end{figure}

\begin{table}
\caption{Estimated $\hat{{\beta}}$ with warfarin data with kernel assisted learning}
    \centering
    \begin{tabular}{c c c c}
        \hline
         Variable& Estimated Parameter & SE & p-value \\
\cmidrule(lr){1-1}  \cmidrule(lr){2-4}
          Intercept &   -0.463   & 0.064 &0.000\\
          Height    &   -0.263   & 0.101 &0.005\\
          Gender    &    0.268   & 0.134 &0.023\\
          VKORC1.AG &   -0.4682  & 0.094 &0.000\\
         \hline
    \end{tabular}
    \label{t5}
\end{table}

%\begin{figure*}
%    \subfloat[]{
%        \includegraphics[width=0.5\textwidth]{Figure2a.eps}
%        }
%    \subfloat[]{
%        \includegraphics[width=0.25\textwidth]{Figure2b.eps}
%        }
%    \caption{Empirical distribution of suggested doses of several methods for the walfarin dataset. Panel (a) presents the distribution of the original doses from the dataset, the distributions of the suggested doses using the kernel assisted learning (KAL), linear outcome-weighted learning (LOL) and kernel based outcome-weighted learning (KOl). Panel (b) presents the histogram of the suggested doses using discretized Q-learning.\label{fig5:suggest}}
%\end{figure*}

\begin{figure}[htp]
    \begin{center}
        \begin{subfigure}[b]{.65\textwidth}
        \centering
        \caption{}
        \includegraphics[width=\textwidth]{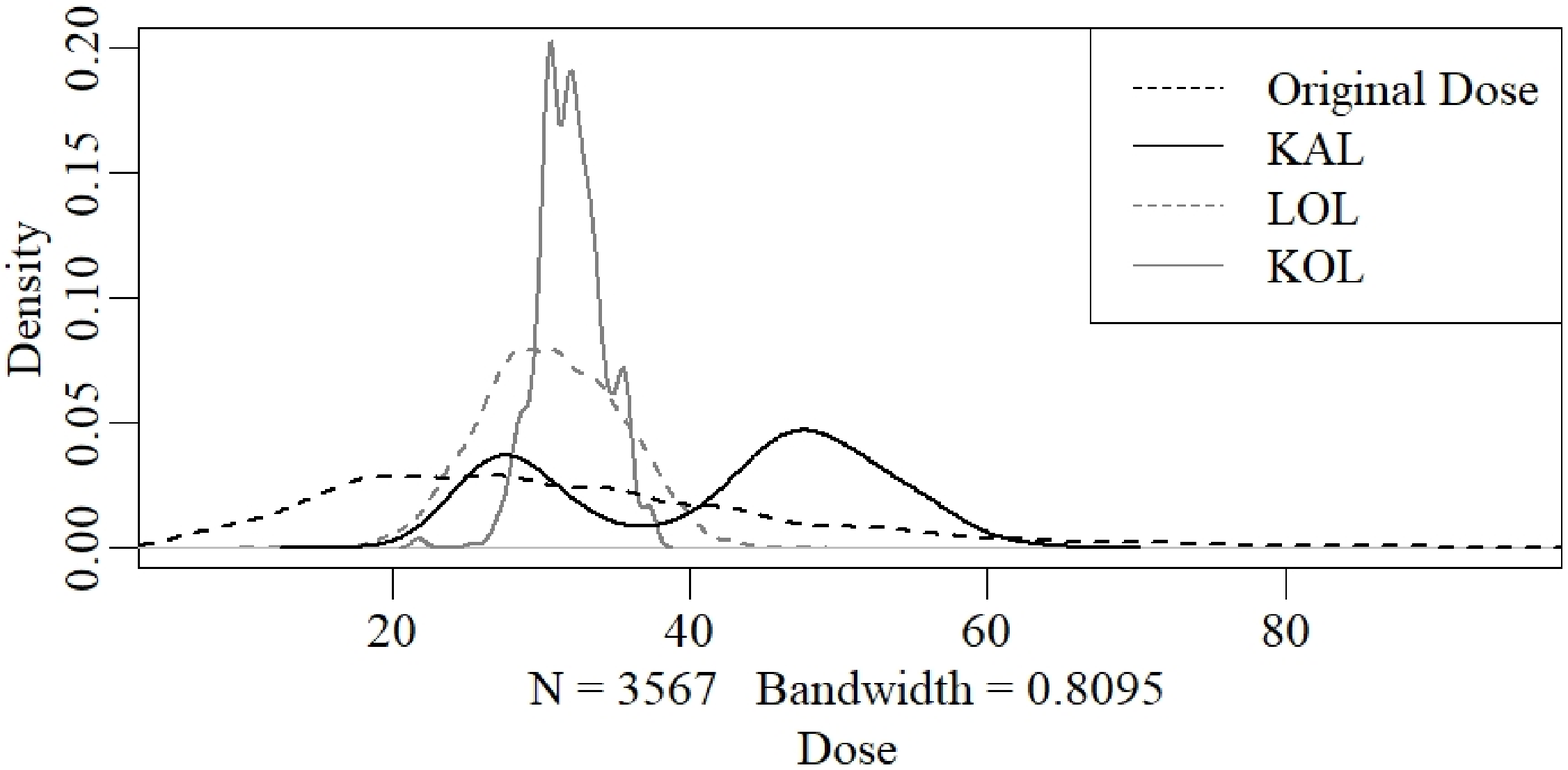}
    \end{subfigure}
        \begin{subfigure}[b]{.28\textwidth}
        \centering
        \caption{}
        \includegraphics[width=\textwidth]{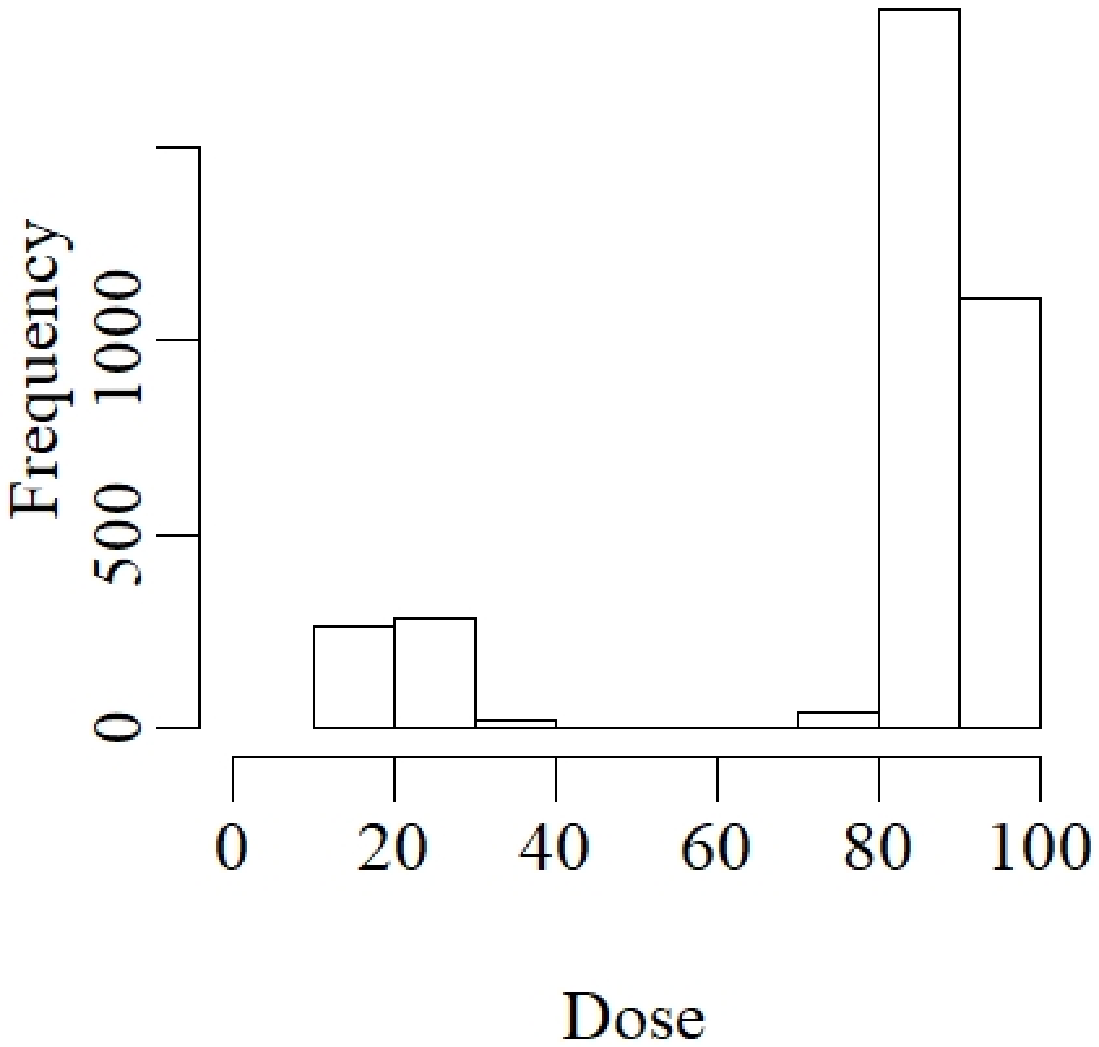}
        \end{subfigure}
    \end{center}
    \caption{Empirical distribution of suggested doses of several methods for the walfarin dataset. In panel (a), the black line is the distribution of the original doses from the dataset. The green line denotes the result from linear O-learning. The blue line denotes the result from kernel based O-learning. The red line denotes the result from kernel assisted learning. Panel (b) is the histogram of the suggested doses using discretized Q-learning.\label{fig5:suggest}}
\end{figure}

\begin{figure*}
    \begin{center}
    \includegraphics[width=0.6\textwidth]{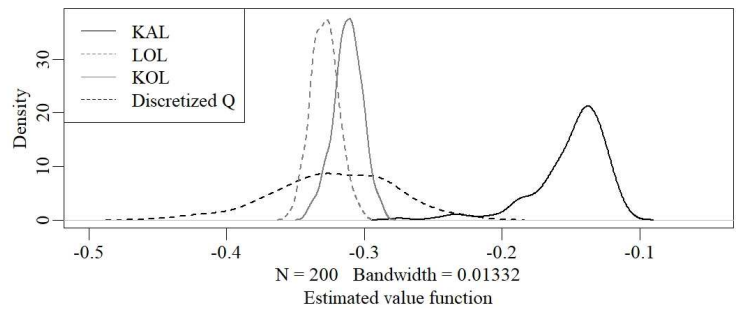}
    \end{center}
    \caption{Empirical distribution of the estimated value function over 200 repetitions for the warfarin dataset.\label{fig6:value}}
\end{figure*}

%-------------------------------------------------------
\section{Discussion and Conclusion}
%-------------------------------------------------------

 The proposed kernel assisted learning method for estimating the optimal individualized dose rule provides the possibility of conducting statistical inference with estimated dose rules, thus providing insights on the importance of the covariates in the dosing process. In our simulation settings, our method was capable of identifying the optimal individualized dose rule when the optimal dose rule was inside the prespecified class of rules. In the warfarin dosing case, based on the three covariates selected, the suggested dose lead to better expected clinical result compared to the other methods. Application of the proposed methodology is not limited to optimal dose finding. This method can also be applied to any scenario where continuous decision making is desired.
 
 The proposed method has several possible extensions. Notice that the form of the prespecified rule class can be extended to a link function with a nonlinear predictor $g\{{\Psi}({X})^T{\beta}\}$ where  ${\Psi}(\cdot)=\{\Psi_1(\cdot),...,\Psi_c(\cdot)\}^T$ are some prespecified basic spline functions and ${\beta}\in\mathbb R^c$ . The accuracy of the approximated value function might also be improved by extending the multivariate kernel ${K_d}({x}/h_x)/h_x$ to $|{H}|^{-1/2}{K_d}({H}^{-1/2}{x})$ \citep{duong2005cross}.  

One weakness of the proposed method is that the accuracy of the estimated value function is sensitive to the choice of bandwidth. The kernel density estimator in the denominator of $M_n({\beta})$ might lead to large bias when the bandwidths are not properly chosen. As the dimension of ${X}$ increases, the choice of the bandwidths is nontrivial. %In this article, we use cross validation by minimizing mean squared error of the Nadaraya-Watson estimator in settings where there are multiple covariates. However, due to the complex form of $M_n({\beta})$, this method might not be optimal when the dimension of the covariates increases. 
The criteria for choosing bandwidths needs to be studied further. 

In the future, we are interested in variable selection when dealing with high dimensional data. Extensions to multi-stage dose finding problems is also of interest. Personalized Dose Finding is still a relatively new problem. With the complicated mechanisms of various diseases, there are many more problems to be tackled in this realm.

\bibliographystyle{ACM-Reference-Format}
\bibliography{paperref}


\begin{thebibliography}{28}

%%% the Web address.

\ifx \showCODEN    \undefined \def \showCODEN     #1{\unskip}     \fi
\ifx \showDOI      \undefined \def \showDOI       #1{#1}\fi
\ifx \showISBNx    \undefined \def \showISBNx     #1{\unskip}     \fi
\ifx \showISBNxiii \undefined \def \showISBNxiii  #1{\unskip}     \fi
\ifx \showISSN     \undefined \def \showISSN      #1{\unskip}     \fi
\ifx \showLCCN     \undefined \def \showLCCN      #1{\unskip}     \fi
\ifx \shownote     \undefined \def \shownote      #1{#1}          \fi
\ifx \showarticletitle \undefined \def \showarticletitle #1{#1}   \fi
\ifx \showURL      \undefined \def \showURL       {\relax}        \fi
% invisible to TeX
\providecommand\bibfield[2]{#2}
\providecommand\bibinfo[2]{#2}
\providecommand\natexlab[1]{#1}
\providecommand\showeprint[2][]{arXiv:#2}

\bibitem[\protect\citeauthoryear{Chen, Zeng, and Kosorok}{Chen
  et~al\mbox{.}}{2016}]%
        {chen2016personalized}
\bibfield{author}{\bibinfo{person}{Guanhua Chen}, \bibinfo{person}{Donglin
  Zeng}, {and} \bibinfo{person}{Michael~R Kosorok}.}
  \bibinfo{year}{2016}\natexlab{}.
\newblock \showarticletitle{Personalized dose finding using outcome weighted
  learning}.
\newblock \bibinfo{journal}{\emph{J. Amer. Statist. Assoc.}}
  \bibinfo{volume}{111}, \bibinfo{number}{516} (\bibinfo{year}{2016}),
  \bibinfo{pages}{1509--1521}.
\newblock


\bibitem[\protect\citeauthoryear{Chen, Fu, He, Kosorok, and Liu}{Chen
  et~al\mbox{.}}{2018}]%
        {chen2018estimating}
\bibfield{author}{\bibinfo{person}{Jingxiang Chen}, \bibinfo{person}{Haoda Fu},
  \bibinfo{person}{Xuanyao He}, \bibinfo{person}{Michael~R Kosorok}, {and}
  \bibinfo{person}{Yufeng Liu}.} \bibinfo{year}{2018}\natexlab{}.
\newblock \showarticletitle{Estimating individualized treatment rules for
  ordinal treatments}.
\newblock \bibinfo{journal}{\emph{Biometrics}} \bibinfo{volume}{74},
  \bibinfo{number}{3} (\bibinfo{year}{2018}), \bibinfo{pages}{924--933}.
\newblock


\bibitem[\protect\citeauthoryear{Chevret}{Chevret}{2006}]%
        {chevret2006statistical}
\bibfield{author}{\bibinfo{person}{Sylvie Chevret}.}
  \bibinfo{year}{2006}\natexlab{}.
\newblock \bibinfo{booktitle}{\emph{Statistical methods for dose-finding
  experiments}}. Vol.~\bibinfo{volume}{24}.
\newblock \bibinfo{publisher}{Wiley Online Library}.
\newblock


\bibitem[\protect\citeauthoryear{Consortium}{Consortium}{2009}]%
        {international2009estimation}
\bibfield{author}{\bibinfo{person}{International Warfarin~Pharmacogenetics
  Consortium}.} \bibinfo{year}{2009}\natexlab{}.
\newblock \showarticletitle{Estimation of the warfarin dose with clinical and
  pharmacogenetic data}.
\newblock \bibinfo{journal}{\emph{New England Journal of Medicine}}
  \bibinfo{volume}{360}, \bibinfo{number}{8} (\bibinfo{year}{2009}),
  \bibinfo{pages}{753--764}.
\newblock


\bibitem[\protect\citeauthoryear{Duong and Hazelton}{Duong and
  Hazelton}{2005}]%
        {duong2005cross}
\bibfield{author}{\bibinfo{person}{Tarn Duong} {and} \bibinfo{person}{Martin~L
  Hazelton}.} \bibinfo{year}{2005}\natexlab{}.
\newblock \showarticletitle{Cross-validation bandwidth matrices for
  multivariate kernel density estimation}.
\newblock \bibinfo{journal}{\emph{Scandinavian Journal of Statistics}}
  \bibinfo{volume}{32}, \bibinfo{number}{3} (\bibinfo{year}{2005}),
  \bibinfo{pages}{485--506}.
\newblock


\bibitem[\protect\citeauthoryear{Fan, Lu, Song, and Zhou}{Fan
  et~al\mbox{.}}{2017}]%
        {fan2017concordance}
\bibfield{author}{\bibinfo{person}{Caiyun Fan}, \bibinfo{person}{Wenbin Lu},
  \bibinfo{person}{Rui Song}, {and} \bibinfo{person}{Yong Zhou}.}
  \bibinfo{year}{2017}\natexlab{}.
\newblock \showarticletitle{Concordance-assisted learning for estimating
  optimal individualized treatment regimes}.
\newblock \bibinfo{journal}{\emph{Journal of the Royal Statistical Society:
  Series B (Statistical Methodology)}} \bibinfo{volume}{79},
  \bibinfo{number}{5} (\bibinfo{year}{2017}), \bibinfo{pages}{1565--1582}.
\newblock


\bibitem[\protect\citeauthoryear{Henderson, Ansell, and Alshibani}{Henderson
  et~al\mbox{.}}{2010}]%
        {henderson2010regret}
\bibfield{author}{\bibinfo{person}{Robin Henderson}, \bibinfo{person}{Phil
  Ansell}, {and} \bibinfo{person}{Deyadeen Alshibani}.}
  \bibinfo{year}{2010}\natexlab{}.
\newblock \showarticletitle{Regret-regression for optimal dynamic treatment
  regimes}.
\newblock \bibinfo{journal}{\emph{Biometrics}} \bibinfo{volume}{66},
  \bibinfo{number}{4} (\bibinfo{year}{2010}), \bibinfo{pages}{1192--1201}.
\newblock


\bibitem[\protect\citeauthoryear{Johnson, Gong, Whirl-Carrillo, Gage, Scott,
  Stein, Anderson, Kimmel, Lee, Pirmohamed, et~al\mbox{.}}{Johnson
  et~al\mbox{.}}{2011}]%
        {johnson2011clinical}
\bibfield{author}{\bibinfo{person}{Julie~A Johnson}, \bibinfo{person}{Li Gong},
  \bibinfo{person}{Michelle Whirl-Carrillo}, \bibinfo{person}{Brian~F Gage},
  \bibinfo{person}{Stuart~A Scott}, \bibinfo{person}{CM Stein},
  \bibinfo{person}{JL Anderson}, \bibinfo{person}{Stephen~E Kimmel},
  \bibinfo{person}{Ming Ta~Michael Lee}, \bibinfo{person}{M Pirmohamed},
  {et~al\mbox{.}}} \bibinfo{year}{2011}\natexlab{}.
\newblock \showarticletitle{Clinical Pharmacogenetics Implementation Consortium
  Guidelines for CYP2C9 and VKORC1 genotypes and warfarin dosing}.
\newblock \bibinfo{journal}{\emph{Clinical Pharmacology \& Therapeutics}}
  \bibinfo{volume}{90}, \bibinfo{number}{4} (\bibinfo{year}{2011}),
  \bibinfo{pages}{625--629}.
\newblock


\bibitem[\protect\citeauthoryear{Kang, Lu, and Zhang}{Kang
  et~al\mbox{.}}{2018}]%
        {kang2018estimation}
\bibfield{author}{\bibinfo{person}{Suhyun Kang}, \bibinfo{person}{Wenbin Lu},
  {and} \bibinfo{person}{Jiajia Zhang}.} \bibinfo{year}{2018}\natexlab{}.
\newblock \showarticletitle{On estimation of the optimal treatment regime with
  the additive hazards model}.
\newblock \bibinfo{journal}{\emph{Statistica Sinica}} \bibinfo{volume}{28},
  \bibinfo{number}{3} (\bibinfo{year}{2018}), \bibinfo{pages}{1539}.
\newblock


\bibitem[\protect\citeauthoryear{Kosorok}{Kosorok}{2008}]%
        {kosorok2008introduction}
\bibfield{author}{\bibinfo{person}{Michael~R Kosorok}.}
  \bibinfo{year}{2008}\natexlab{}.
\newblock \bibinfo{booktitle}{\emph{Introduction to empirical processes and
  semiparametric inference.}}
\newblock \bibinfo{publisher}{Springer}.
\newblock


\bibitem[\protect\citeauthoryear{Laber and Zhao}{Laber and Zhao}{2015}]%
        {laber2015tree}
\bibfield{author}{\bibinfo{person}{EB Laber} {and} \bibinfo{person}{YQ Zhao}.}
  \bibinfo{year}{2015}\natexlab{}.
\newblock \showarticletitle{Tree-based methods for individualized treatment
  regimes}.
\newblock \bibinfo{journal}{\emph{Biometrika}} \bibinfo{volume}{102},
  \bibinfo{number}{3} (\bibinfo{year}{2015}), \bibinfo{pages}{501--514}.
\newblock


\bibitem[\protect\citeauthoryear{Liang, Lu, and Song}{Liang
  et~al\mbox{.}}{2018}]%
        {liang2018deep}
\bibfield{author}{\bibinfo{person}{Shuhan Liang}, \bibinfo{person}{Wenbin Lu},
  {and} \bibinfo{person}{Rui Song}.} \bibinfo{year}{2018}\natexlab{}.
\newblock \showarticletitle{Deep advantage learning for optimal dynamic
  treatment regime}.
\newblock \bibinfo{journal}{\emph{Statistical theory and related fields}}
  \bibinfo{volume}{2}, \bibinfo{number}{1} (\bibinfo{year}{2018}),
  \bibinfo{pages}{80--88}.
\newblock


\bibitem[\protect\citeauthoryear{Murphy}{Murphy}{2003}]%
        {murphy2003optimal}
\bibfield{author}{\bibinfo{person}{Susan~A Murphy}.}
  \bibinfo{year}{2003}\natexlab{}.
\newblock \showarticletitle{Optimal dynamic treatment regimes}.
\newblock \bibinfo{journal}{\emph{Journal of the Royal Statistical Society:
  Series B (Statistical Methodology)}} \bibinfo{volume}{65},
  \bibinfo{number}{2} (\bibinfo{year}{2003}), \bibinfo{pages}{331--355}.
\newblock


\bibitem[\protect\citeauthoryear{Qian and Murphy}{Qian and Murphy}{2011}]%
        {qian2011performance}
\bibfield{author}{\bibinfo{person}{Min Qian} {and} \bibinfo{person}{Susan~A
  Murphy}.} \bibinfo{year}{2011}\natexlab{}.
\newblock \showarticletitle{Performance guarantees for individualized treatment
  rules}.
\newblock \bibinfo{journal}{\emph{Annals of statistics}} \bibinfo{volume}{39},
  \bibinfo{number}{2} (\bibinfo{year}{2011}), \bibinfo{pages}{1180}.
\newblock


\bibitem[\protect\citeauthoryear{Rich, Moodie, and Stephens}{Rich
  et~al\mbox{.}}{2014}]%
        {rich2014simulating}
\bibfield{author}{\bibinfo{person}{Benjamin Rich}, \bibinfo{person}{Erica~EM
  Moodie}, {and} \bibinfo{person}{David~A Stephens}.}
  \bibinfo{year}{2014}\natexlab{}.
\newblock \showarticletitle{Simulating sequential multiple assignment
  randomized trials to generate optimal personalized warfarin dosing
  strategies}.
\newblock \bibinfo{journal}{\emph{Clinical Trials}} \bibinfo{volume}{11},
  \bibinfo{number}{4} (\bibinfo{year}{2014}), \bibinfo{pages}{435--444}.
\newblock


\bibitem[\protect\citeauthoryear{Robins}{Robins}{2004}]%
        {robins2004optimal}
\bibfield{author}{\bibinfo{person}{James~M Robins}.}
  \bibinfo{year}{2004}\natexlab{}.
\newblock \showarticletitle{Optimal structural nested models for optimal
  sequential decisions}. In \bibinfo{booktitle}{\emph{Proceedings of the second
  seattle Symposium in Biostatistics}}. Springer, \bibinfo{pages}{189--326}.
\newblock


\bibitem[\protect\citeauthoryear{Rubin}{Rubin}{1978}]%
        {rubin1978bayesian}
\bibfield{author}{\bibinfo{person}{Donald~B Rubin}.}
  \bibinfo{year}{1978}\natexlab{}.
\newblock \showarticletitle{Bayesian inference for causal effects: The role of
  randomization}.
\newblock \bibinfo{journal}{\emph{The Annals of statistics}}
  (\bibinfo{year}{1978}), \bibinfo{pages}{34--58}.
\newblock


\bibitem[\protect\citeauthoryear{Schuster et~al\mbox{.}}{Schuster
  et~al\mbox{.}}{1969}]%
        {schuster1969estimation}
\bibfield{author}{\bibinfo{person}{Eugene~F Schuster} {et~al\mbox{.}}}
  \bibinfo{year}{1969}\natexlab{}.
\newblock \showarticletitle{Estimation of a probability density function and
  its derivatives}.
\newblock \bibinfo{journal}{\emph{The Annals of Mathematical Statistics}}
  \bibinfo{volume}{40}, \bibinfo{number}{4} (\bibinfo{year}{1969}),
  \bibinfo{pages}{1187--1195}.
\newblock


\bibitem[\protect\citeauthoryear{Shi, Fan, Song, and Lu}{Shi
  et~al\mbox{.}}{2018}]%
        {shi2018high}
\bibfield{author}{\bibinfo{person}{Chengchun Shi}, \bibinfo{person}{Alin Fan},
  \bibinfo{person}{Rui Song}, {and} \bibinfo{person}{Wenbin Lu}.}
  \bibinfo{year}{2018}\natexlab{}.
\newblock \showarticletitle{High-dimensional a-learning for optimal dynamic
  treatment regimes}.
\newblock \bibinfo{journal}{\emph{Annals of statistics}} \bibinfo{volume}{46},
  \bibinfo{number}{3} (\bibinfo{year}{2018}), \bibinfo{pages}{925}.
\newblock


\bibitem[\protect\citeauthoryear{Shi, Song, and Lu}{Shi et~al\mbox{.}}{2019}]%
        {shi2019concordance}
\bibfield{author}{\bibinfo{person}{Chengchun Shi}, \bibinfo{person}{Rui Song},
  {and} \bibinfo{person}{Wenbin Lu}.} \bibinfo{year}{2019}\natexlab{}.
\newblock \showarticletitle{Concordance and value information criteria for
  optimal treatment decision}.
\newblock \bibinfo{journal}{\emph{Annals of Statistics}}
  (\bibinfo{year}{2019}).
\newblock


\bibitem[\protect\citeauthoryear{Song, Luo, Zeng, Zhang, Lu, and Li}{Song
  et~al\mbox{.}}{2017}]%
        {song2017semiparametric}
\bibfield{author}{\bibinfo{person}{Rui Song}, \bibinfo{person}{Shikai Luo},
  \bibinfo{person}{Donglin Zeng}, \bibinfo{person}{Hao~Helen Zhang},
  \bibinfo{person}{Wenbin Lu}, {and} \bibinfo{person}{Zhiguo Li}.}
  \bibinfo{year}{2017}\natexlab{}.
\newblock \showarticletitle{Semiparametric single-index model for estimating
  optimal individualized treatment strategy}.
\newblock \bibinfo{journal}{\emph{Electronic journal of statistics}}
  \bibinfo{volume}{11}, \bibinfo{number}{1} (\bibinfo{year}{2017}),
  \bibinfo{pages}{364}.
\newblock


\bibitem[\protect\citeauthoryear{Watkins and Dayan}{Watkins and Dayan}{1992}]%
        {watkins1992q}
\bibfield{author}{\bibinfo{person}{Christopher~JCH Watkins} {and}
  \bibinfo{person}{Peter Dayan}.} \bibinfo{year}{1992}\natexlab{}.
\newblock \showarticletitle{Q-learning}.
\newblock \bibinfo{journal}{\emph{Machine learning}} \bibinfo{volume}{8},
  \bibinfo{number}{3-4} (\bibinfo{year}{1992}), \bibinfo{pages}{279--292}.
\newblock


\bibitem[\protect\citeauthoryear{Xiao, Zhang, and Lu}{Xiao
  et~al\mbox{.}}{2019}]%
        {xiao2019robust}
\bibfield{author}{\bibinfo{person}{Wei Xiao}, \bibinfo{person}{Hao~Helen
  Zhang}, {and} \bibinfo{person}{Wenbin Lu}.} \bibinfo{year}{2019}\natexlab{}.
\newblock \showarticletitle{Robust regression for optimal individualized
  treatment rules}.
\newblock \bibinfo{journal}{\emph{Statistics in medicine}}
  \bibinfo{volume}{38}, \bibinfo{number}{11} (\bibinfo{year}{2019}),
  \bibinfo{pages}{2059--2073}.
\newblock


\bibitem[\protect\citeauthoryear{Zhang, Tsiatis, Davidian, Zhang, and
  Laber}{Zhang et~al\mbox{.}}{012a}]%
        {zhang2012estimating}
\bibfield{author}{\bibinfo{person}{Baqun Zhang}, \bibinfo{person}{Anastasios~A
  Tsiatis}, \bibinfo{person}{Marie Davidian}, \bibinfo{person}{Min Zhang},
  {and} \bibinfo{person}{Eric Laber}.} \bibinfo{year}{2012a}\natexlab{}.
\newblock \showarticletitle{Estimating optimal treatment regimes from a
  classification perspective}.
\newblock \bibinfo{journal}{\emph{Stat}} \bibinfo{volume}{1},
  \bibinfo{number}{1} (\bibinfo{year}{2012a}), \bibinfo{pages}{103--114}.
\newblock


\bibitem[\protect\citeauthoryear{Zhang, Tsiatis, Laber, and Davidian}{Zhang
  et~al\mbox{.}}{012b}]%
        {zhang2012robust}
\bibfield{author}{\bibinfo{person}{Baqun Zhang}, \bibinfo{person}{Anastasios~A
  Tsiatis}, \bibinfo{person}{Eric~B Laber}, {and} \bibinfo{person}{Marie
  Davidian}.} \bibinfo{year}{2012b}\natexlab{}.
\newblock \showarticletitle{A robust method for estimating optimal treatment
  regimes}.
\newblock \bibinfo{journal}{\emph{Biometrics}} \bibinfo{volume}{68},
  \bibinfo{number}{4} (\bibinfo{year}{2012b}), \bibinfo{pages}{1010--1018}.
\newblock


\bibitem[\protect\citeauthoryear{Zhao, Kosorok, and Zeng}{Zhao
  et~al\mbox{.}}{2009}]%
        {zhao2009reinforcement}
\bibfield{author}{\bibinfo{person}{Yufan Zhao}, \bibinfo{person}{Michael~R
  Kosorok}, {and} \bibinfo{person}{Donglin Zeng}.}
  \bibinfo{year}{2009}\natexlab{}.
\newblock \showarticletitle{Reinforcement learning design for cancer clinical
  trials}.
\newblock \bibinfo{journal}{\emph{Statistics in medicine}}
  \bibinfo{volume}{28}, \bibinfo{number}{26} (\bibinfo{year}{2009}),
  \bibinfo{pages}{3294--3315}.
\newblock


\bibitem[\protect\citeauthoryear{Zhao, Zeng, Rush, and Kosorok}{Zhao
  et~al\mbox{.}}{2012}]%
        {zhao2012estimating}
\bibfield{author}{\bibinfo{person}{Yingqi Zhao}, \bibinfo{person}{Donglin
  Zeng}, \bibinfo{person}{A~John Rush}, {and} \bibinfo{person}{Michael~R
  Kosorok}.} \bibinfo{year}{2012}\natexlab{}.
\newblock \showarticletitle{Estimating individualized treatment rules using
  outcome weighted learning}.
\newblock \bibinfo{journal}{\emph{J. Amer. Statist. Assoc.}}
  \bibinfo{volume}{107}, \bibinfo{number}{499} (\bibinfo{year}{2012}),
  \bibinfo{pages}{1106--1118}.
\newblock


\bibitem[\protect\citeauthoryear{Zhou, Zhang, Lu, and Li}{Zhou
  et~al\mbox{.}}{2019}]%
        {zhou2019restricted}
\bibfield{author}{\bibinfo{person}{Jie Zhou}, \bibinfo{person}{Jiajia Zhang},
  \bibinfo{person}{Wenbin Lu}, {and} \bibinfo{person}{Xiaoming Li}.}
  \bibinfo{year}{2019}\natexlab{}.
\newblock \showarticletitle{On restricted optimal treatment regime estimation
  for competing risks data}.
\newblock \bibinfo{journal}{\emph{Biostatistics}} (\bibinfo{year}{2019}).
\newblock


\end{thebibliography}


\appendix

%-------------------------------------------------------
\section{Proof of Theorem \ref{theorem1}}
%-------------------------------------------------------

We first prove the uniform convergence of $M_n({\beta})$ to $M({\beta})$. For simplicity of notation, let's define: 

    $
    m_{ x}( x,a)= \partial m( x,a)/{\partial  x}
    $,
    $
    m_{2 x}( x,a)=\partial m_2( x,a)/{\partial  x}
    $,
    $
    m_{2a}( x,a)=\partial m_2( x,a)/{\partial a}
    $.
    Similarly, $f_a({x},a)=\partial f_{X,A}({x},a)/ {\partial a} $, $f_x({x},a)=\partial f_{X,A}({x},a)/{\partial {x}}$. We write $M_n({\beta})$ as 
    $$
    M_n({\beta})=\int_x \frac{A_n({x};{\beta})}{B_n({x};{\beta})}C_n({x})d{x},
    $$
where 
\begin{gather*}
    A_n({x};{\beta})=\frac{1}{n} \sum_{i=1}^n Y_i  \frac{1}{h_x^q}{K_q}(\frac{{x}-{X}_i}{h_x})\frac{1}{h_a} K\big\{\frac{g({\beta}^T{x})-A_i}{h_a}\big\},\\
    B_n({x};{\beta})=\frac{1}{n} \sum_{i=1}^n    \frac{1}{h_x^q}{K_q}(\frac{{x}-{X}_i}{h_x}) \frac{1}{h_a} K\big\{\frac{g({\beta}^T{x})-A_i}{h_a}\big\},\\
    C_n({x})=\frac{1}{n} \sum_{i=1}^n  \frac{1}{h_x^q}{K_q}(\frac{{x}-{X}_i}{h_x}).
\end{gather*}
Notice that $M({\beta})$ can be written as
$$M({\beta})=\int_{{x}}\frac{A({x};{\beta})}{B({x};{\beta})}C({\beta})d{x},$$  
where $A({x};{\beta})=m\{{x},g({\beta}^T{x})\}f_{X,A}\{{x},g({\beta}^T{x})\}$,  $B({x};{\beta})=f_{X,A}\{{x},g({\beta}^T{x})\}$ and $C({x})=f_X({x})$.
Thus,
\begin{gather*}
    \sup_{{\beta}}\Big|M_n({\beta})-M({\beta})\Big|
    =
    \sup_{{\beta}}\Big| \int_x\Big\{\frac{A_n({x};{\beta})}{B_n({x};{\beta})}C_n({x})-\frac{A({x};{\beta})}{B({x};{\beta})}C({x})\Big\} d{x} \Big|\\
    \leq \sup_{{\beta}} \Big|\int_x \Big\{ \frac{A_n({x};{\beta})}{B_n({x};{\beta})}-\frac{A({x};{\beta})}{B({x};{\beta})}\Big\}C_n({x})d{x}\Big|+\sup_{{\beta}} \Big|\int_x \frac{A({x};{\beta})}{B({x};{\beta})}\Big\{C_n({x})-C({x})\Big\} d{x}\Big|\\
    \leq \sup_{a,{x}}\Big|\frac{A^*_n({x},a)}{B^*_n({x},a)}-\frac{A^*({x},a)}{B^*({x},a)}\Big|\Big\{\int_x C_n({x})d{x}\Big\}+\sup_{a,{x}} |m\big({x},a)| \Big\{\int_x \Big|C_n({x})-C({x})\Big|d{x}\Big\},
\end{gather*}
where 
\begin{gather*}
    A_n^*({x},a)=\frac{1}{n}\sum_{i=1}^n Y_i\frac{1}{h_x^q}{K_q}(\frac{{x}-{X}_i}{h_x})\frac{1}{h_a} K(\frac{a-A_i}{h_a}),\\
    B^*_n({x},a)=\frac{1}{n}\sum_{i=1}^n \frac{1}{h_x^q}{K_q}(\frac{{x}-{X}_i}{h_x})\frac{1}{h_a} K(\frac{a-A_i}{h_a}),
\end{gather*} and $A^*({x},a)=m({x},a)f_{X,A}({x},a)$,  $B^*({x},a)=f_{X,A}({x},a)$. It is trivial to prove that $\int_x C_n({x})d{x}=\int_x C({x})d{x}=1$. Thus, under the boundedness of $m({x},a)$, we only need to show that:
\begin{gather*}
    \int_{{x}} |C_n({x})-C({x})|d{x}=o_p(1),\\
    \sup_{a,x}\Big|\frac{A^*_n({x},a)}{B^*_n({x},a)}-\frac{A^*({x},a)}{B^*({x},a)}\Big|\to 0 =o_p(1).
\end{gather*}
To prove the first equation, notice that $\int_{{x}} |C_n({x})-C({x})|d{x}\le \int_{{x}} C_n({x})d{x}+\int_{{x}}C({x})d{x}=2$. By the dominated convergence theorem, it suffices to show the uniform convergence of the kernel density estimate $C_n({x})$ to $C({x})$, which can be proved according to Schuster (1969).

For the second equation,
\begin{gather*}
    \sup_{a,{x}}\Big|\frac{A^*_n({x},a)}{B^*_n({x},a)}-\frac{A^*({x},a)}{B^*({x},a)}\Big|\\
    =\sup_{a,{x}} \Big|\frac{\big\{A_n^*({x},a)-A^*({x},a)\big\}B^*({x},a)-A^*({x},a)\big\{B_n^*({x},a)-B^*({x},a)\big\}}{B_n^*({x},a)B^*({x},a)}\Big|\\
    \leq \sup_{{x},a}\Big|\frac{A_n^*({x},a)-A^*({x},a)}{B_n^*({x},a)}\Big|+\sup_{{x},a}\Big|\frac{\big\{B_n^*({x},a)-B^*({x},a)\big\}A^*({x},a)}{B^*({x},a)B^*_n({x},a)}\Big|.
\end{gather*}
Under the boundedness of $A^*({x},a)$ and the assumption that $f_{X,A}({x},a)$ is uniformly bounded away from 0, it suffices to show that
    $
    \sup_{a,x}\big|A_n^*({x},a)-A^*({x},a)\big|=o_p(1)$, and
    $\sup_{a,x}\big|B_n^*({x},a)-B^*({x},a)\big|=o_p(1)$.

To prove the uniform convergence of $A_n^*({x},a)$, notice that: 
\begin{equation}\label{A*twoparts}
    \sup_{a,x}\Big|A_n^*({x},a)-A^*({x},a)\Big|\le \sup_{a,x}\Big|A_n^*({x},a)-E\big\{A_n^*({x},a)\big\}\Big|+\sup_{a,x}\Big|E\big\{A_n^*({x},a)\big\}-A^*({x},a)\Big|.
\end{equation}
We prove the convergence of the two parts on the right side separately. First we obtain,
\begin{align*}
    &E\big\{ A_n^*({x},a)\big\}= \\
    &\int_{{x}_i,a_i}\int_{y_i} \big\{y_i\frac{1}{h_x^q}{K_q}(\frac{{x}-{x}_i}{h_x})\frac{1}{h_a} K(\frac{a-a_i}{h_a}) \big\}    f_{Y|X,A}(y_i|{x_i},a_i)f_{X,A}({x}_i,a_i)dy_i d{x}_i da_i\\
    &= \int_{{x}_i,a_i}  \big\{ m({x}_i,a_i)\frac{1}{h_x^q}{K_q}(\frac{{x}-{x}_i}{h_x})\frac{1}{h_a} K(\frac{a-a_i}{h_a})  \big\} f_{X,A}({x}_i,a_i) d{x}_i da_i\\
    &= \int_{{u}} \int_{v=(a-1)/{h_a}}^{a/{h_a}} m({x}-{u}h_x,a-vh_a) {K_q}({u}) K(v) f_{X,A}({x}-{u}h_x,a-vh_a) dv d{u}\\
    &=\int_{{u}} \int_{v=(a-1)/{h_a}}^{a/{h_a}} \big\{m\big({x},a\big)-{u}h_x m_x\big({x},a\big)-vh_am_a\big({x},a\big)+O(h_x^2)+O(h_xh_a)+O(h_a^2)\big\}\\&K({u})K(v)
    \big\{f_{X,A}\big({x},a\big)-{u}h_x f_x \big({x},a\big)- vh_af_a\big({x},a\big)+O(h_x^2)+O(h_xh_a)+O(h_a^2)\big\} dv d{u}\\
    &= m\big({x},a\big)f_{X,A}\big({x},a\big)+ O(h_x^2)+O(h_a^2)\\
    &=A^*({x},a)+ O(h_x^2)+O(h_a^2),
\end{align*}
where the third equality is achieved by letting ${u}=({x}-{X})/{h_x}$ and $v=\{g({\beta}^T{x})-A\}/{h_a}$. The fourth equation is from Taylor expansion and the fifth equation is based on Assumption (A1) that $\int_u {u} {K_q}({u})d {u}=0$, $\int_{\{g({\beta}' x)-1\}/{h_a}}^{g({\beta}' x)/{h_a}} K(v)dv=1-O(h_a^2)$ and $\int_{\{g({\beta}' x)-1\}/{h_a}}^{g({\beta}' x)/{h_a}} v K(v)dv=O(h_a)$. By the assumption of uniform boundedness of the second order derivatives of $m({x},a)$ and $f_{X,A}({x},a)$, we have $\sup_{x,a}|E\big\{ A_n^*({x},a)\big\}-A^*({x},a) |\to 0$. 

Then we prove the convergence of the first part of Equation (\ref{A*twoparts}):
\begin{align*}
    &\sup_{{x},a}\Big|A_n^*({x},a)-E\big\{A_n^*({x},a)\big\}\Big|\\
    &=\frac{1}{h_x^q h_a}\sup_{{x},a}\Big|\int_{{X},A}m({X},A){K_q}(\frac{{x}-{X}}{h_x})K(\frac{a-A}{h_a})d\big\{F_n({X},A)-F({X},A)\big\}\Big|\\
    &=\frac{1}{h_x^q h_a}\sup_{{x},a}\Bigg|\Big[\big\{F_n({x},a)-F({x},a)\big\}m({X},A){K_q}(\frac{{x}-{X}}{h_x})K(\frac{a-A}{h_a})\Big]_{-\infty}^{\infty}\\
    &-\int_{{X},A}\big\{F_n({x},a)-F({x},a)\big\}d\big\{ m({X},A){K_q}(\frac{{x}-{X}}{h_x})K(\frac{a-A}{h_a})\big\}\Bigg|\\
    &\leq\frac{C_3}{h_x^q h_a}\sup_{{x},a}\Big|F_n({x},a)-F({x},a)\Big|,
\end{align*}
where $C_3$ is a constant, $F(X,A)$ is the cumulative joint distribution of $X$ and $A$ and $F_n(x,a)=\{\sum_{i=1}^n I(X_i\le x,A_i\le a)\}/n$. Here $X_i\le x$ means that each term of $X_i$ is less than or equal to the corresponding term of $x$. The last inequality can be obtained by the boundedness of $m({x},a)$, ${K_q}({u})$ and $K(v)$.
  By Lemma 2.1 of Schuster (1969) we know that there exists a universal constant $C_4$ such that for any $n>0$, $\epsilon_n>0$, $P_F\Big\{\sup_{{x},a}\big|F_n({x},a)-F({x},a)\big|>\epsilon\Big\}\leq C_4\exp (-2n\epsilon^2)$. For $n$ sufficiently large:
\begin{gather*}
    P\Big\{\sup_{{x},a}\big|A_n^*({x},a)-A^*({x},a)\big|>\epsilon\Big\}\leq P\Big\{\sup_{{x},a}\Big|A_n^*({x},a)-E\big\{A_n^*({x},a)\big\}\Big|>\frac{\epsilon}{2}\Big\}\\
    \leq P\Big\{\sup_{{x},a}\big|F_n({x},a)-F({x},a)\big|>\frac{h_x^q h_a\epsilon}{2C_3}\Big\}\leq C_4\exp\big(-2n\frac{h_x^{2q}h_a^2\epsilon^2}{C_3^2}\big).
\end{gather*}
If $nh_x^{2q}h_a^2\to \infty$ then $ P\Big\{\sup_{{x},a}\big|A_n^*({x},a)-A^*({x},a)\big|\Big\}\to 0$. %$\sum_{n=1}^\infty exp(-2nh_x^{2q}h_a^2\epsilon^2/C_3^2)$ is finite and $\sum_{n=1}^\infty P\{\sup_{x,a}|A_n^*(x,a)-A^*(x,a)|>\epsilon\}$ is finite, then by Borel-cantelli lemma, $P\{\sup_{x,a}|A_n^*(x,a)-A^*(x,a)|>\epsilon\}\to 0$. 
Thus the uniform convergence of $A_n^*({x},a)$ is proved. The uniform convergence of $B_n^*({x},a)$ can be proved similarly. Thus, we can obtain $\sup_{{\beta}\in\Theta}\big|M_n({x};{\beta})-M({x};{\beta})\big|\xrightarrow{p}0 $. By Theorem 2.10 of Kosorok (2008), we now obtain that $\hat{{\beta}}_n\xrightarrow{p}{\beta}^*$.
%-------------------------------------------------------

%-------------------------------------------------------
\section{Proof of Theorem \ref{theorem2}}
%-------------------------------------------------------
Since $\hat{{\beta}}_n$ and $ \beta^*$ are maximizers of $M_n( \beta)$ and $M( \beta)$, they are solutions of $S_n( \beta)=0$ and $S( \beta)=0$, where, $ {S}({\beta})=\partial M({\beta})/{\partial {\beta}}$ and ${S}_n({\beta})=\partial  M_n({\beta})/ {\partial {\beta}}$. By Taylor expansion, we have
\begin{gather*}
    0={S}_n(\hat{{\beta}}_n)={S}_n({{\beta}}^*)+{D}_n ({{\beta}}^*)(\hat{{\beta}}_n- {{\beta}}^*)+\frac{1}{2}(\hat{{\beta}}_n-{{\beta}}^*)^T\frac{\partial^2}{\partial {{\beta}} \partial {{\beta}}^T} {S}_n(\tilde {{\beta}}) (\hat {{\beta}}_n-{{\beta}}^*)\\
    ={S}_n({{\beta}}^*)+\Big\{{D}_n ({{\beta}}^*)+\frac{1}{2}(\hat{{\beta}}_n-{{\beta}}^*)^T\frac{\partial^2}{\partial {{\beta}} \partial {{\beta}}^T} {S}_n({\tilde{{\beta}}})\Big\} (\hat {{\beta}}_n-{{\beta}}^*),
\end{gather*} where $\tilde{{\beta}}$ is on the line segment connecting  $\hat{{\beta}_n}$ and ${\beta}^*$, ${D}_n({{\beta}})=\partial^2  M_n({{\beta}}) / (\partial {{\beta}}\partial {{\beta}}^T)$, ${D}({{\beta}})=\partial^2  M({{\beta}}) / (\partial {{\beta}}\partial {{\beta}}^T)$. To prove the weak convergence of $\hat{ \beta}_n$, we can first prove that:
\begin{equation}\label{S}
    (n h_x^q h_a^3)^{1/2}\big\{{S}_n({\beta}^*)-{S}({\beta}^*)\big\}\to N\big\{0,{\Sigma}_S({\beta}^*)\big\},
\end{equation} in distribution as $n\to\infty$, 
\begin{equation}\label{D}
    {D}_n ({\beta}^*)-{D}({\beta}^*)=o_p(1),
\end{equation}
\begin{equation}\label{residual}
    \frac{1}{2}(\hat{{\beta}}_n-{\beta}^*)^T\frac{\partial^2}{\partial {\beta} \partial {\beta}^T} {S}_n(\tilde{{\beta}})=o_p(1).
\end{equation}
Then we obtain:
\begin{align*}
    &(n h_x^q h_a^3)^{1/2} (\hat {{\beta}}_n-{{\beta}}^*)=-\Big\{{D}_n ({{\beta}}^*)+\frac{1}{2}(\hat{{\beta}}_n-{{\beta}}^*)^T\frac{\partial^2}{\partial {{\beta}} \partial {{\beta}}^T} {S}_n({{\beta}}^*)\Big\} (n h_x^q h_a^3 )^{1/2} {S}_n({{\beta}}^*)\\
    &\to N\Big\{0,{D}({\beta}^*)^{-1}{\Sigma}_S({\beta}^*){D}({\beta}^*)^{-1}\Big\}
\end{align*}
in distribution as $n\to\infty$.

%--------------
\subsection{Proof of Equation \ref{S}}
%--------------
To prove Equation \ref{S}, we first write ${S}_n({\beta})$ as:
\begin{align*}
    {S}_n({\beta})&=\frac{\partial M_n({\beta})}{\partial {\beta}}=\int_x\frac{\tilde {{ A}}_n({x};{\beta})B_n({x};{\beta})   -A_n({x};{\beta})\tilde {{ B}}_n({x};{\beta})}{B_n({x};{\beta})^2}C_n({x})d{x},
\end{align*} where
    \begin{align*}
    \tilde {{ A}}_n({x};{\beta}) &= \frac{\partial}{\partial{\beta}} A_n({x};{\beta})= \frac{1}{n}\sum_{i=1}^n Y_i \frac{1}{h_x^q}{K_q}(\frac{{x}-{X}_i}{h_x})\frac{1}{h_a^2}\dot K\big\{\frac{g({\beta}^T{x})-A_i}{h_a}\big\}\dot g({\beta}^T{x}){x},\\
     \tilde {{ B}}_n({x};{\beta}) &= \frac{\partial}{\partial{\beta}}B_n({x};{\beta})= \frac{1}{n}\sum_{i=1}^n \frac{1}{h_x^q} {K_q}(\frac{{x}-{X}_i}{h_x})\frac{1}{h_a^2}\dot K\big\{\frac{g({\beta}^T{x})-A_i}{h_a}\big\}\dot g({\beta}^T{x}){x}.
\end{align*}

Since ${S}_n({\beta})$ is of the integration form, to calculate the limiting distribution of ${S}_n$, we can first calculate the limiting distribution of the part inside the integral for a fixed ${x}$. %$\sqrt{nh_x^q h_a^3} \tilde {{A}}_n({x};{\beta}^*)$ for a fixed ${x}$.

Let the parts inside the integral of $ S_n({\beta})$ and $ S({\beta})$ be:
\begin{gather*}
    {G}_n({x};{\beta})=\frac{\tilde {{A}}_n({x};{\beta}) B_n({x};{\beta})-A_n({x};{\beta})\tilde{ {B}}_n({x};{\beta})}{B_n({x};{\beta})^2}C_n({x}),\\
    {G}({x};{\beta})=\frac{\tilde{ {A}}({x};{\beta}) B({x};{\beta})-A({x};{\beta})\tilde {{B}}({x};{\beta})}{B^2({x};{\beta})}C({x}),
\end{gather*}
where 
$
    \tilde{{A}}({x};{\beta})=\Big[m\big\{{x},g({\beta}^T{x})\big\}f_a\big\{{x},g({\beta}^T{x})\big\}+m_a\big\{{x},g({\beta}^T{x})\big\}f_{X,A}\big\{{x},g({\beta}^T{x})\big\}\Big] \dot g({\beta}^T{x}){x}$ and \\
    $
    \tilde{{B}}({x};{\beta})=f_a\{{x},g\big({\beta}^T{x})\big\}$.

To prove the limit distribution of ${G}_n({\beta})-{G}({\beta})$, we need the following lemma:
%-------------------------------------------------------
\begin{lemma}\label{lemma3}
    If $\{A_n\}_{n=1}^{\infty}$ and $\{B_n\}_{n=1}^{\infty}$ are two sequences of random variables and $c_n(A_n-A)\to N(0,\Sigma _A)$ in distribution and $d_n(B_n-B)\to N(0, \Sigma_B)$ in distribution, where $c_n/d_n\to 0$ as $n\to \infty$. Then:
\begin{gather*}
    c_n(A_nB_n-AB)=c_n(A_n-A)B+o_p(1).
\end{gather*}
\end{lemma} 
%-------------------------------------------------------
\begin{proof}
    Notice that:
$
    c_n(A_nB_n-AB)=(c_n/{d_n})A_n\{d_n(B_n-B)\}+c_n(A_n-A)B,
$
where $d_n(B_n-B)$ converges to a normal distribution, $A_n$ converge in probability to $A$ and $c_n/d_n\to 0$. Thus the first term is $o_p(1)$. Then we have $c_n(A_nB_n-AB)=c_n(A_n-A)B+o_p(1)$.
\end{proof}

Under the assumption of the boundedness of the first three derivatives of $m({x},a)$ and $f({x},a)$, we can prove that $E\big\{\tilde {{A}}_n({x};{\beta}) \big\} =\tilde {{A}}({x};{\beta})+{O}(h_x^2+h_a^2)$. Together with the law of large numbers, we obtain that $\tilde{{A}}_n({x};{\beta})\xrightarrow{p}\tilde{{A}}({x};{\beta})$.
Since $\tilde {{A}}_n({x};{\beta})$ is the sum of $n$ i.i.d variables, with the central limit theorem we can obtain that $(n h_x^q h_a^3)^{1/2}\{\tilde {{A}}_n({x};{\beta})-\tilde {{A}}({x};{\beta})\}$ converges to a normal distribution if $nh_x^qh_a^3 \text{Var}\{\tilde {{A}}_n({x};{\beta})\}$ converges to a constant covariance matrix. Notice now that, 
 \begin{align*}
    &\text{Var}\big\{\tilde {{A}}_n({x};{\beta})\big\}
    =\frac{1}{n} \text{Var}\Big[ Y_i \frac{1}{h_x^q}{K_q}(\frac{{x}-{X}_i}{h_x})\frac{1}{h_a^2}\dot K\big\{\frac{g({\beta}^T{x})-A_i}{h_a} \big\}\dot g({\beta}^T{x}){x} \Big]\\
    &=\frac{1}{n}\dot g^2({\beta}^T{x}){xx}^T\Bigg\{E \Big[Y_i^2 \frac{1}{h_x^{2q}}{K_q}^2\big(\frac{{x}-{X}_i}{h_x}\big)\frac{1}{h_a^4}\dot K^2\big\{\frac{g({\beta}^T{x})-A_i}{h_a}\big\} \Big]-E^2\big[\tilde{{A}}_n({x};{\beta})\big]\Bigg\} \\
    &=\frac{1}{n}\dot g^2({\beta}^T{x}){xx}^T E \Big[Y_i^2 \frac{1}{h_x^{2q}}{K_q}^2(\frac{{x}-{X}_i}{h_x})\frac{1}{h_a^4}\dot K^2\big\{\frac{g({\beta}^T{x})-A_i}{h_a}\big\} \Big] +{O}(\frac{1}{n}),
\end{align*}
where the expectation in the last equation can be calculated similarly as before:
\begin{align*}
     &E\Big[Y_i^2 \frac{1}{h_x^{2q}}{K_q}^2(\frac{{x}-{X}_i}{h_x})\frac{1}{h_a^4}\dot K^2\big\{\frac{g({\beta}^T{x})-A_i}{h_a}\big\} \Big]\\
     &=\frac{1}{h_x^q h_a^3}\Big[m_2\big\{{x},g({\beta}^T{x})\big\}f_{X,A}\big\{{x},g({\beta}^T{x})\big\}{\kappa}_{0,2}\dot \kappa_{0,2}+O(h_x^2)+O(h_a^2)+O(h_xh_a)\Big].
\end{align*}
Thus,
$$
    nh_x^q h_a^3\text{Var}\big\{\tilde {{A}}_n({x};{\beta})\big\}=\dot g^2({\beta}^T{x}){xx}^T m_2\big\{{x},g({\beta}^T{x})\big\}f_{X,A}\big\{{x},g({\beta}^T{x})\big\}{\kappa}_{0,2}\dot \kappa_{0,2}+{O}(h_x^qh_a^3).
$$
Therefore, for $nh_x^q h_a^3\to \infty$ as $n\to \infty$, we have:
\begin{gather*}
   (n h_x^q h_a^3)^{1/2}\big\{\tilde {{A}}_n({x};{\beta})-\tilde {{A}}({x};{\beta})\big\}\to\\
   N\Big[0,g^2({\beta}^T{x}){xx}^T m_2\big\{{x},g({\beta}^T{x})\big\} f_{X,A}\big\{{x},g({\beta}^T{x})\big\}{\kappa}_{0,2}\dot \kappa_{0,2}\Big]
\end{gather*}
in distribution as $n\to\infty$.

Similarly, we can obtain that, as $n\to\infty$, 
$$(n h_x^q h_a^3)^{1/2}\big\{\tilde {{B}}_n({x};{\beta})- \tilde {{B}}({x};{\beta})\big\}\xrightarrow{d}N\Big[0, \dot g^2({\beta}^T{x}){xx}^T f_{X,A}\big\{{x},g({\beta}^T{x})\big\}{\kappa}_{0,2}\dot \kappa_{0,2}\Big],$$

$$(n h_x^q h_a^3)^{1/2}\big\{A_n({x};{\beta})-A({x};{\beta})\big\}\xrightarrow{d}N\Big[0,m_2\big\{{x},g({\beta}^T{x})\big\} f_{X,A}\big\{{x},g({\beta}^T{x})\big\} {\kappa}_{0,2}\tilde \kappa_{0,2}\Big],$$ 
$$(n h_x^q h_a^3)^{1/2}\big\{B_n({x};{\beta})-B({x};{\beta})\big\}\xrightarrow{d}N\Big[0,f_{X,A}\big\{{x},g({\beta}^T{x})\big\} {\kappa}_{0,2}\tilde \kappa_{0,2}\Big],$$ where $\tilde \kappa_{0,2}=\int K^2(s)ds $. 

By Lemma \ref{lemma3} and the above convergence results, we obtain that:
\begin{align*}
    &(n h_x^q h_a^3)^{1/2}\big\{\tilde {{A}}_n({x};{\beta})B_n({x};{\beta})-\tilde {{B}}_n({x};{\beta})A_n({x};{\beta}) \big\}\\
    &=(n h_x^q h_a^3)^{1/2}\big\{\tilde {{A}}_n({x};{\beta})B({x};{\beta})-\tilde {{B}}_n({x};{\beta})A({x};{\beta})\big\}+{o}_p(1)\\
    &=(n h_x^q h_a^3)^{1/2}\big\{\frac{1}{n}\sum_{i=1}^n{\Phi}_i({x};{\beta})\big\} +{o}_p(1),
\end{align*}
where
$$
    {\Phi}_i({x};{\beta})=\big\{Y_iB({x};{\beta})-A({x};{\beta})\big\} \frac{1}{h_x^q}{K_q}(\frac{{x}-{X}_i}{h_x})\frac{1}{h_a^2}\dot K\big\{\frac{g({\beta}^T{x})-A_i}{h_a}\big\}\dot g({\beta}^T{x}){x}.
$$ Similar to previous calculations, by the central limit theorem we can prove that \\ $(n h_x^q h_a^3)^{1/2}$ $\{\sum_{i=1}^n{\Phi}_i({x};{\beta})/n\}$ converge to a normal distribution, where the covariance of the asymptotic distribution is ${\Sigma}_{{\Phi}}({x};{\beta})$:
$$
    {\Sigma_{\Phi}}({x};{\beta})=\dot g^2({\beta}^T{x}){xx}^T {\kappa}_{0,2}\dot \kappa_{0,2} \Big[m_2\big\{{x},g({x}'{\beta})\big\}-m^2\big\{{x},g({x}'{\beta})\big\}\Big]f^3_{X,A}\big\{{x},g({x}'{\beta})\big\}.
$$

Notice that
\begin{gather*}
(n h_x^q h_a^3)^{1/2}{G}_n({x};{\beta})
=(n h_x^q h_a^3)^{1/2}  \Big\{\frac{1}{n}\sum_{i=1}^n{\Phi}_i({x};{\beta})\Big\}\frac{C_n({x})}{B_n^2({x};{\beta})}+{o}_p(1).
\end{gather*}
Together with $C_n({x})\xrightarrow{p}C({x})$, $B_n({x};{\beta})\xrightarrow{p}B({x};{\beta})$, and Slusky's theorem, we now obtain that: 
\begin{gather*}
    (n h_x^q h_a^3)^{1/2}\{{G}_n({x};{\beta})-{G}({x};{\beta})\}\to N\big\{{0},{\Sigma_G}({x};{\beta})\big\},
\end{gather*}
where:
$$
   % &=\dot g^2({\beta}^T{x}){xx}^T {\kappa}_{0,2}\dot \kappa_{0,2} C^2({x}) \frac{m_2\big({x},g({\beta}^T{x})\big)B^2({x};{\beta})-A^2({x};{\beta})}{B^3({x};{\beta})}\\
    {\Sigma_G}({x};{\beta})=\dot g^2({\beta}^T{x}){xx}^T {\kappa}_{0,2}\dot \kappa_{0,2} f^2_X({x})\frac{m_2\big\{{x},g({x}'{\beta})\big\}-m^2\big\{{x},g({x}'{\beta})\big\}}{f_{X,A}\big\{{x},g({x}'{\beta})\big\}}.
$$
Now let us calculate the covariance of ${S}_n({x};{\beta})$.
By the tightness of ${G}_n({x};{\beta})$ and ${G}({x};{\beta})$, \\ $(n h_x^q h_a^3)^{1/2}\{{G}_n({x};{\beta})-{G}({x};{\beta})\}$ converges weakly to a Gaussian process ${\mathcal G}({x})$with mean 0 and covariance function  ${\Sigma_\mathcal G}({\beta})$, where ${\Sigma}_{{\mathcal G}}({x}_1,{x}_2;{\beta})$ is the limit of 
\begin{align*}
    &\text{Cov}\big\{(n h_x^q h_a^3)^{1/2}{G}_n({x}_1;{\beta}),(n h_x^q h_a^3)^{1/2}{G}_n({x}_2;{\beta})\big\}\\
    &=\frac{h_x^qh_a^3}{n}\sum_{i=1}^n\sum_{j=1}^n \text{Cov}\big\{{\Phi}_i({x}_1;{\beta}),{\Phi}_j({x}_2;{\beta})\big\}\frac{C({x}_1)C({x}_2)}{B^2({x}_1;{\beta})B^2({x}_2;{\beta})}+{o}_p(1)\\
    &=\frac{h_x^qh_a^3}{n}\sum_{i=1}^n \text{Cov}\big\{{\Phi}_i({x}_1;{\beta}),{\Phi}_i({x}_2;{\beta})\big\}\frac{C({x}_1)C({x}_2)}{B^2({x}_1;{\beta})B^2({x}_2;{\beta})}+{o}_p(1)\\
    &=h_x^qh_a^3\text{Cov}\big\{{\Phi}_1({x}_1;{\beta}),{\Phi}_1({x}_2;{\beta})\big\}\frac{C({x}_1)C({x}_2)}{B^2({x}_1;{\beta})B^2({x}_2;{\beta})}+{o}_p(1)\\
    &={T}({x}_1,{x}_2;{\beta})\big\{ \int_{{u}} {K_q}({u}){K_q}({u}+\frac{{x}_2-{x}_1}{h_x}) d{u}\big\} \Big[\int_v\dot K(v) \dot K\big\{v+\frac{g({\beta}^T{x}_1)-g({\beta}^T{x}_2)}{h_a}\big\}dv\Big]
    \\&+{o}_p(1),
\end{align*}
for ${x}_1, {x}_2\in \mathcal X$, where 
\begin{align*}
    &{T}({x}_1,{x}_2;{\beta})=\\
    &\Big[ m_2\big\{{x}_1,g({\beta}^T{x}_1)\big\}B({x}_1;{\beta})B({x}_2;{\beta})-m\big\{{x}_1,g({\beta}^T{x}_1)\big\}  \big\{  A({x}_2;{\beta})B({x}_1;{\beta})+
    \\&A({x}_1;{\beta})B({x}_2;{\beta})  \big\}  +A({x}_1;{\beta})A({x}_2;{\beta})\Big]\dot g({\beta}^T{x}_1)\dot g({\beta}^T{x}_2){x}_1{x}_2^T\frac{C({x}_1)C({x}_2)}{B^2({x}_1;{\beta})B^2({x}_2;{\beta})}.
\end{align*}

When ${x}_1\neq {x}_2$, $\int_{{u}} {K_q}({u})K\{{u}+({x}_2-{x}_1)/{h_x}\} d{u}$ and $\int_v\dot K(v) \dot K[v+$ $\{g({\beta}^T{x}_1)-$ $g({\beta}^T{x}_2)\}$ $/{h_a}]dv$ will converge to 0 as $h_x,h_a\to 0$. Thus ${\Sigma_\mathcal G}({x}_1,{x}_2;{\beta})={0}$ for ${x}_1\neq {x}_2$.
Therefore, 
\begin{equation*}
    (n h_x^q h_a^3)^{1/2} \big\{{S}_n({\beta}^*)-{S}({\beta}^*)\big\} = \int_{{x}}(n h_x^q h_a^3)^{1/2} \big\{{G}_n({x};{\beta}^*)-{G}({x};{\beta}^*)\big\}dx \to N\big\{{0},{\Sigma_S}({\beta}^*)\big\},
\end{equation*}
in distribution, where $
    {\Sigma_S}({\beta})=\int_{{x}_1}\int_{{x}_2} {\Sigma_{\mathcal G}}({x}_1,{x}_2;{\beta})d{x}_1d{x}_2=\int_{{x}} {\Sigma_G}({x};{\beta})d{x}.$

%--------------------------------------
\subsubsection{Proof of Equation (\ref{D})}
First, write ${D}_n({\beta})$ as:
\begin{align*}
    {D}_n({\beta})
    &=\frac{\partial^2}{\partial {\beta}\partial {\beta}^T}M_n({\beta})\\
    &=\int_x \Big\{\frac{\frac{\partial^2}{\partial{\beta}\partial{\beta}^T}A_n({x};{\beta})}{B_n({x};{\beta})}
    -2\frac{\frac{\partial}{\partial{\beta}}A_n({x};{\beta})\frac{\partial}{\partial{\beta}}B_n({x};{\beta})}{B_n^2({x};{\beta})}
    -\frac{A_n({x};{\beta})\frac{\partial^2}{\partial{\beta}\partial{\beta}^T}B_n({x};{\beta})}{B_n({x};{\beta})^2}\\
    &+2\frac{A_n({x};{\beta})\frac{\partial}{\partial{\beta}}B_n({x};{\beta})\frac{\partial}{\partial{\beta}^T}B_n({x};{\beta})}{B_n({x};{\beta})^3}\Big\} C_n({x}) d{x}.
\end{align*}
Similar to previous calculations, under the assumption of boundedness of the first three order derivatives of $m({x},a)$ and $f_{X,A}({x},a)$, we obtain that  $\partial^2 A_n({x};{\beta})/(\partial {\beta}\partial{\beta}^T)$ converges in probability to:

\begin{align*}
    &2\Big[{m}_{aa}\big\{{x},g({\beta}^T{x})\big\}f_{X,A}\big({x},g({\beta}^T{x})\big)+
    {m}_a\big\{{x},g({\beta}^T{x})\big\}{f}_a\big\{{x},g({\beta}^T{x})\big\}+\\&
    m\big\{{x},g({\beta}^T{x})\big\}{f}_{aa}\big\{{x},g({\beta}^T{x})\big\}\Big]\dot g^2({\beta}^T{x}){xx}^T\\
    &
    + \Big[{m}_a\big\{{x};g({\beta}^T{x})\big\} f_{X,A}\big\{{x};g({\beta}^T{x})\big\} + 
    m\big\{{x},g({\beta}^T{x})\big\} {f}_a\big\{{x};g({\beta}^T{x})\big\}\Big]\ddot g({\beta}^T{x}){xx}^T,
\end{align*}
and $\partial^2 B_n({x};{\beta})/ ({\partial {\beta}\partial{\beta}^T} )$ converges in probability to 
\begin{align*}
    2 {f}_{aa}\big\{{x},g({\beta}^T{x})\big\}\dot g^2({\beta}^T{x}){xx}^T+
    {f}_a\big\{{x},g({\beta}^T{x})\big\}\ddot g({\beta}^T{x}){xx}^T.
\end{align*}

Together with the previous convergence results for $\tilde{{A}}_n({x};{\beta})$, $\tilde{{B}}_n({x};{\beta})$, $A_n({x};{\beta})$, $B_n({x};{\beta})$, $C_n({x})$, we obtain that $D_n(\beta)$ converge in probability to
\begin{gather*}
    \int_x \Big[ {m}_{aa} \big\{{x},g({\beta}^T{x})\big\}\dot g^2({\beta}^T{x}) +{m}_a\big\{{x},g({\beta}^T{x})\big\}\ddot g({\beta}^T{x})\Big]f_X({x}) {x x}^T d{x}={D}({\beta}).
\end{gather*}

%-------------------------------
\subsubsection{Proof of Equation (\ref{residual})}
For notation, let $x_l$ be the $l$ th component of the vector $\textbf x$, and $\beta_j$ and $\beta_k$ are the $j$ th and $k$ th component of vector ${\beta}$, $j,k,l\in \{1,...d\}$. Let $S_{n,l}({\beta})$ be the $l$ th component of vector ${S}_n({\beta})$. Since we have proved that $\hat{{\beta}}_n-{\beta}^*$ converge in probability to $0$, to prove Equation (\ref{residual}) it suffices to show that:
$
        \partial^2{S}_{n,l}({\beta})/({\partial{\beta} \partial{\beta}}^T)={O}_p(1),
$
By calculation, for $j,k,l\in \{1,...,d\}$: %${\beta}\in \Theta$:
\begin{align*}
        &\frac{\partial^2}{\partial\beta_j\partial\beta_k}S_{n,l}({\beta})\\&
        =\int_{{x}}  \Bigg\{  \frac{\frac{\partial^2}{\partial\beta_j\partial\beta_k}\tilde A_{n,l}({x};{\beta})}{B_n({x};{\beta})}
        -3\frac{\tilde B_{n,k}({x};{\beta})\frac{\partial}{\partial\beta_j}\tilde A_{n,l}({x};{\beta})}{B_n^2({x};{\beta})}
        -3\frac{\tilde A_{n,l}({x};{\beta})\frac{\partial}{\partial \beta_k}\tilde B_{n,j}({x};{\beta})}{B_n^2({x};{\beta})}\\&
        -\frac{A_n({x};{\beta})\frac{\partial^2}{\partial\beta_j\partial\beta_k}\tilde B_{n,l}({x};{\beta})}{B_n^2}
        +6\frac{\tilde A_{n,l}({x};{\beta})\tilde B_{n,j}({x};{\beta}) \tilde B_{n,k}({x};{\beta})}{B_n^3}\\&
        +2\frac{A_n({x};{\beta})\big(\frac{\partial\tilde B_{n,l}}{\partial \beta_k}\tilde B_{n,j} + \frac{\partial\tilde B_{n,j}}{\partial \beta_k}\tilde B_{n,l} +\frac{\partial\tilde B_{n,l}}{\partial \beta_j}\tilde B_{n,k}\big)}{B_n^3}\\&
        -6\frac{A_n({x};{\beta})\tilde B_{n,l}({x};{\beta})\tilde B_{n,j}({x};{\beta}) \tilde B_{n,k}({x};{\beta})}{B_n^4({x};{\beta})} \Bigg\}
        C_n({x})d{x},
    \end{align*}
where $\tilde A_{n,j}({x};{\beta})$ is the $j$ th component of vector $\tilde {{A}}_n({x};{\beta})$ and $\tilde {{B}}_{n,j}({x};{\beta})$ is the $j$ th component of vector $\tilde B_n({x};{\beta})$. With similar calculation as before, under the assumption that the first four orders of derivatives of $m({x},a)$ and $f_{X,A}({x},a)$ are bounded, we obtain that: $\partial^2S_{n,l}({\beta})/({\partial\beta_j\partial\beta_k})={O}_p(1)$. Thus $
        \partial^2 {S}_{n,l}({\beta})/({\partial{\beta} \partial{\beta}}^T)={O}_p(1).
$
%---------------------------------
\section{Estimation of Covariance} 
%---------------------------------
From above, the covariance of the asymptotic distribution for $(n h_x^q h_a^3)^{1/2}(\hat{{\beta}}_n-{\beta}^*)$ is given by: $${D}({\beta}^*)^{-1}{\Sigma}_{{S}}({\beta}^*){D}({\beta}^*)^{-1}.$$ First, ${D}({\beta})$ can be estimated with ${D}_n({{\beta}})$. Then for the estimation of ${\Sigma}_{{S}}({\beta})$, notice that $
    {\Sigma_S}({\beta})=\int_{{x}_1}\int_{{x}_2} {\Sigma_{\mathcal G}}({x}_1,{x}_2;{\beta})d{x}_1d{x}_2=\int_{{x}} {\Sigma_G}({x};{\beta})d{x}$, and 
    $$
    (n h_x^q h_a^3)^{1/2}{G}_n({x};{\beta})
=(n h_x^q h_a^3)^{1/2}  \Big\{\frac{1}{n}\sum_{i=1}^n{\Phi}_i({x};{\beta})\Big\}\frac{C_n({x})}{B_n^2({x};{\beta})}+{o}_p(1).
$$ 
Therefore, ${\Sigma}_{{G}}({x;\beta})={\Sigma}_\Phi(x;\beta)C_n^2(x)/B_n^4(x;\beta)$ where $\Sigma_{\Phi} (x;\beta)$ can be estimated empirically by the sample covariance of  ${\Phi}_i({x};{\beta})$. ${\Phi}_i({x};{\beta})$ is approximated by plugging in the $A_n({x};{\beta})$, $B_n({x};{\beta})$, $C_n({x})$ for $A({x};{\beta})$, $B({x};{\beta})$ and $C({x})$. Finally, we plug in $\hat{{\beta}}_n$ for ${\beta}^*$ to obtain the estimated covariance.
\end{document}